\documentclass[a4paper]{article}

\usepackage[utf8]{inputenc}
\usepackage[dvipsnames]{xcolor}
\usepackage{mathtools}
\usepackage{cmll}
\usepackage{mdframed}
\usepackage{subcaption}
\usepackage{tikz}
\usepackage{stmaryrd}
\usepackage{ebproof}
\usepackage{graphicx}
\usepackage{hyperref}
\usepackage{macros-HAL}
\usepackage{amsthm,amssymb,amsmath}
\usepackage{cleveref}

\newtheorem{theorem}{Theorem}
\newtheorem{proposition}[theorem]{Proposition}
\newtheorem{lemma}[theorem]{Lemma}
\newtheorem{corollary}[theorem]{Corollary}

\theoremstyle{definition}
\newtheorem{definition}[theorem]{Definition}

\theoremstyle{remark}
\newtheorem{remark}[theorem]{Remark}
\newtheorem{notation}[theorem]{Notation}
\newtheorem{example}[theorem]{Example}

\newcommand\hypo{\Hypo}
\newcommand\infer{\Infer}
\newcommand\ellipsis{\Ellipsis}

\tikzset{
    dot/.style = {fill=black, circle, outer sep=2.5pt, inner sep=1pt},
    every loop/.style={}
}

\title{Stellar Resolution: Multiplicatives} 

\author{Boris Eng \\ Université Sorbonne Paris Nord \\ LIPN -- UMR 7030\\ \texttt{engboris@hotmail.fr} \and Thomas Seiller\\ CNRS \\ LIPN -- UMR 7030\\ \texttt{seiller@lipn.fr}}

\begin{document}

\maketitle

\begin{abstract}
We present a new asynchronous model of computation named \textit{Stellar Resolution} based on first-order unification \cite{herbrand1930recherches, robinson1965machine}. This model of computation is obtained as a formalisation of Girard's transcendental syntax programme, sketched in a series of three articles \cite{girard2017transcendental,girard2016transcendental,girard2016transcendental3}. As such, it is the first step towards a proper formal treatment of Girard's proposal to tackle first-order logic in a proofs-as-program approach \cite{girard2016transcendental3}. After establishing formal definitions and basic properties of stellar resolution, we explain how it generalises traditional models of computation, such as logic programming and combinatorial models such as Wang tilings. We then explain how it can represent multiplicative proof-structures \cite{girard1987linear}, their cut-elimination and the correctness criterion of Danos-Regnier \cite{danos1989structure}. Further use of realisability techniques lead to dynamic semantics for Multiplicative Linear Logic, following previous Geometry of Interaction models.
\end{abstract}

\section{Introduction}\label{sec:intro}

We present a new asynchronous model of computation named \textit{Stellar Resolution} based on first-order unification \cite{herbrand1930recherches, robinson1965machine}. This model arises from work in proof theory, and more precisely proof-theoretic semantics related to the Curry-Howard correspondence between proofs and programs \cite{curry1934functionality, howard1980formulae}. While Curry-Howard traditionally assimilates terms in the $\lambda$-calculus with proofs in intuitionnistic logic, stellar resolution terms -- named constellations -- will be assimilated with (generalisations of) proofs in linear logic.

Linear Logic was introduced by Girard  \cite{girard1987linear} as a refinement of intuitionnistic logic inspired from semantics of $\lambda$-calculus \cite{girard1988normal}. Soon after the introduction of linear logic, the \emph{geometry of interaction} programme emerged \cite{girard1989towards}, aiming at defining semantics of proofs and programs accounting for the dynamics of the cut-elimination procedure. This \emph{dynamic semantics} approach, a major inspiration behind game semantics \cite{hyland2000full,abramsky2000full}, distinguishes itself from denotational semantics in which cut-elimination is represented as equality.

However, the geometry of interaction ambition went beyond the dynamic semantics aspects. It also aimed to reconstruct logic -- in particular linear logic -- from what could be understood as an untyped model of computation by using realisability techniques. In some way, the idea is the same as the reconstruction of simple types from pure $\lambda$-calculus. In $\lambda$-calculus, one defines an orthogonality relation between terms $t$ and contexts $E(\cdot)$ by letting $t\perp E(\cdot)\Leftrightarrow E(t) \text{ is strongly normalising}$. A type $A$ is then defined as a set of contexts $T_A$ -- understood as tests --, and a term $t$ is considered typable with the type $A$ when $t\in T_A^\bot$, i.e. $\forall E(\cdot)\in T_A, t\perp E(\cdot)$. We refer the reader to, e.g., Riba's work\cite{riba2007strong} for more details.

In this aspect of reconstructing linear logic, several geometry of interaction models were defined using operators algebras \cite{girard1989geometry,girard1988geometry,girard2011geometry}, unification algebras \cite{girard1995geometry}, graphs \cite{seiller2012interaction,seiller2016interaction} and graphings \cite{seiller2017interaction,seiller2019interaction,seiller2016interaction2}. Although all these models did define rich models, that were in particular used to study computational complexity \cite{baillot2001elementary,aubert2016characterizing,aubert2016logarithmic,seiller2018interaction,seiller2020probabilistic}, they failed with respect to two different aspects. Firstly, the objects used to interpret even the most basic proofs were most of the time infinite objects and even when they were not, types were defined through an infinite number of tests. Secondly, the obtained models did interpret soundly the fragments of linear logic considered, but no completeness results exist\footnote{While this aspect is a failure somehow, it is also a feature as the models are very rich and open other paths of reflexion.}.

Recently, Girard published a series of articles \cite{girard2013three, girard2017transcendental, girard2016transcendental, girard2016transcendental3} sketching the main lines of a new kind of model that would have the qualities of geometry of interaction models, but improve on them at least concerning the first failure mentioned above. Those articles are interesting and claim great improvements, proposing in particular a Curry-Howard interpretation of first-order logic \cite{girard2016transcendental3}. However, these articles are too inexact in form to serve satisfactorily as the basis of a mathematical theory\footnote{The formulation is borrowed from Church's critic of von Mises notion of kollektiv \cite{church1940}.}. The current work is the first step towards a proper formal account of the model, with underlying motivation the representation of first-order logic and its possible applications in relation with descriptive complexity results \cite{immerman2012descriptive}, such as the Immerman-Vardi theorem \cite{vardi1982complexity,immermann1986relational}.

\noindent\textbf{Contributions.}
We formally describe a model of computation named stellar resolution, which extends the model of computation vaguely described by Girard. We prove the main properties of the model. In particular, while Girard claimed the failure of the Church-Rosser property, we are able to prove it holds for stellar resolution (Theorem \ref{thm:associativity}).  We also relate it to standard models of computation such as Wang tilings \cite{wang1961proving} and abstract tiles assembly models \cite{winfree1998algorithmic,patitz2014introduction}, which have applications in bio-computing \cite{seeman1982nucleic,woods2019diverse}. We then explain how this model captures the dynamics of cut-elimination for MLL, the multiplicative fragment of linear logic (Theorem \ref{def:dynamics}), and the correctness criterion for proof-structures -- a syntax for MLL (Theorem \ref{thm:correctness2}). Lastly, we explain how realisability techniques similar to those used in $\lambda$-calculus can be used to define types that organise into a denotational semantics for MLL (a $\ast$-autonomous category), and prove soundness and completeness of the model w.r.t. MLL+MIX, an extension of MLL with the so-called MIX rule (Theorems \ref{thm:soundness} \& \ref{thm:completeness}).

\section{Stellar Resolution}\label{sec:stellar}

\subsection{First-order Terms and Unification}\label{subsec:unifcation}

Stellar resolution is based on the theory of unification. We here recall basic definitions and refer the reader to the article of \textit{Lassez et al.} \cite{lassez1988unification} for more details.

\begin{definition}[Signature]
\label{def:signature}
A signature $\mathbb{S} = (\mathcal{V}, \mathcal{F}, \mathtt{ar})$ consists of an infinite countable set $\mathcal{V}$ of variables, a countable set $\mathcal{F}$ of function symbols whose arities are given by $\mathtt{ar} : \mathcal{F} \rightarrow \nat$.
\end{definition}

\begin{definition}[First-order term]
\label{def:term}
Let $\mathbb{S} = (\mathcal{V}, \mathcal{F}, \mathtt{ar})$ be a signature. The set of terms $\mathcal{T}_\mathbb{S}$ is inductively defined by the grammar $t, u ::= x \mid f(t_1, ..., t_n)$ with $x \in \mathcal{V}, f \in \mathcal{F}, \mathtt{ar}(f) = n$.
\end{definition}

We now fix a signature $\mathbb{S} = (\mathcal{V}, \mathcal{F}, \mathtt{ar})$ until the end of this section.

\begin{definition}[Substitution]
\label{def:subst}\label{def:rename}
A substitution is a function $\theta: \mathcal{V} \rightarrow \mathcal{T}_\mathbb{S}$ extended from variables to terms by $\theta(f(u_1, ..., u_k)) = f(\theta u_1, ..., \theta u_k)$.

A \emph{renaming} is a substitution $\alpha$ such that $\alpha(x)\in\mathcal{V}$ for all $x\in\mathcal{V}$.
\end{definition}

\begin{definition}[Unification problem]
\label{def:unifprob}\label{def:solprob}
An equation is an unordered pair $t \doteq u$ of terms in $\mathcal{T}_\mathbb{S}$. A unification problem is a set $P=\{t_1 \doteq u_1, ..., t_n \doteq u_n\}$ of equations.

A solution for $P$ is a substitution $\theta$ such that for all equations $t_i \doteq u_i$ in $P$, $\theta t_i = \theta u_i$.
\end{definition}

\begin{definition}[Matching]
\label{def:matching}
Two terms $t$ and $u$ are matchable if there exists a renaming $\alpha$ such that $\{\alpha t \doteq u\}$ has a solution.
\end{definition}

\begin{theorem}[Unification algorithm]
\label{def:unifalgo}
The problem of deciding if a solution to a given unification problem $P$ exists is decidable. Moreover, there exists a minimal solution $\unifalgo{P}$ w.r.t the preorder $\theta \preceq \psi \Leftrightarrow \exists \theta'.\psi = \theta' \circ \theta$, unique up to renaming.
\end{theorem}

Let us note that several algorithms were designed to compute the unique solution when it exists, such that the Martelli-Montanari unification algorithm \cite{martelli1982efficient}.

\subsection{Stars and Constellations}\label{subsec:stars}

The stellar resolution model is based on first-order unification but extends it the notions of \emph{polarities} and \emph{colours}. Intuitively, polarised terms are prefixed with either a $+$ or $-$ sign indicating their ability to be composed (with a term of opposite polarity). Colours allow for the consideration of various types of composition. We define the core objects of our model.

\begin{definition}[Coloured Signature]
\label{def:colouredsignature}
A \emph{coloured signature} is a 4-tuple $\mathbb{S} = (\mathcal{V}, \mathcal{F}, \mathcal{C}, \mathtt{ar})$ where $ (\mathcal{V}, \mathcal{F}, \mathtt{ar})$ is a signature and $\mathcal{C} \subseteq \mathcal{F}$ is a specified set of \emph{colours}. The set of terms $\mathcal{T}_\mathbb{S}$ is defined as the set of terms over the signature $(\mathcal{V}, \mathcal{F}, \mathtt{ar})$.
\end{definition}

We will now work with the coloured signature $\mathbb{S} = (\mathcal{V}, \mathcal{F}, \mathcal{C}, \mathtt{ar})$ unless specified otherwise.

\begin{definition}[Ray]
\label{def:ray}
A ray is a term in the grammar $r ::= +c(t_1, ..., t_n) \mid -c(t_1, ..., t_n) \mid t$, where $\{t_1, ..., t_n\} \subseteq \mathcal{T}_\mathbb{S}$ and $c \in \mathcal{C}$ with $\mathtt{ar}(c) = n$. The underlying term of a ray is defined by $\floor{+c(t_1, ..., t_n)} = \floor{-c(t_1, ..., t_n)} = c(t_1, ..., t_n)$ and $\floor{t}=t$.
\end{definition}

\begin{notation}
\label{not:ray}
We will sometimes write $+c.t$ (resp. $-c.t$) instead of $+c(t)$ (resp. $-c(t)$) for unary colours when it is convenient.
\end{notation}

\begin{definition}[Star]
\label{def:star}
A star $\sigma$ over a coloured signature $\mathbb{S}$ is a finite and non-empty multiset of rays, i.e. a finite set $\abs{\sigma}$ together with a map $\rho_\sigma:\abs{\sigma}\rightarrow\rays{\mathbb{S}}$. The set of variables appearing in $\sigma$ is written $\mathtt{vars}(\sigma)$. Stars are written as multisets $\sigma=\gstar{r_1, ..., r_n}$.
\end{definition}

\begin{definition}[Substitutions and $\alpha$-equivalence]
\label{def:alphaeq}
Given a substitution $\theta$, its action extends to rays by letting $\theta(\pm c(t_1, ..., t_n))=\pm\theta(c(t_1, ..., t_n))$ with $\pm \in \{+, -\}$. It also extends to stars: $\theta\gstar{r_1, ..., r_n} = \gstar{\theta r_1, ..., \theta r_n}$.

We say that two stars $\sigma_1, \sigma_2$ are \emph{$\alpha$-equivalent}, written $\sigma_1 \alphaeq \sigma_2$, when there exists a renaming $\alpha$ such that $\alpha \sigma_1 = \sigma_2$.
\end{definition}

\begin{notation}
In this work, stars will be considered up to $\alpha$-equivalence. We therefore define $\setofstars$ as the set of all stars over a coloured signature $\mathbb{S}$, quotiented by $\alphaeq$.
\end{notation}

\begin{definition}[Constellation]
\label{def:constellation}
A constellation $\Sigma$ is a (countable) multiset of stars, i.e. a countable (possibly infinite) set $\abs{\Sigma}$ together with a map $\alpha_\Sigma: \abs{\Sigma}\rightarrow\setofstars$. The variables appearing in $\Sigma$ are considered (sometimes implicitly) bound to their star i.e $\bigcap_{e\in\abs{\Sigma}} \mathtt{vars}(\alpha_\Sigma(e)) = \emptyset$. The set of rays of $\Sigma$ is defined as $\rays{\Sigma}=\{ (s,r) \mid s\in\abs{\Sigma},r\in\abs{\alpha_\Sigma(s)}\}$. A finite constellation will sometimes be written $\Sigma = a_1\sigma_1 + ... + a_n\sigma_n$ ($a_i\in\nat$).
\end{definition}

\begin{example}
\label{ex:constellation}
Here is an example of two finite and an infinite constellation:
\begin{itemize}
    \item $\Sigma_{ex} = \gstar{g(x), f(x), +a(f(x))}+\gstar{-a(y), +b(y)}+\gstar{x, -b(g(x))}+\gstar{+b(x), x}$
    \item $\Sigma_a = \gstar{-a(x), +a(x), -b(x)}$
    \item $\Sigma_\nat$ is defined by $|\Sigma_\nat| = \nat$ and $\alpha_{\Sigma_\nat}(n) = \gstar{-nat.s^n(0), +nat.s^{n+1}(0)}$
\end{itemize}
where $s^0(t) = t$ and $s^{n+1}(t) = s(s^n(t))$, $nat$ is a colour and $s$ and $0$ are function symbols.
\end{example}

\subsection{Diagrams and Reduction}\label{subsec:diagrams}

Constellations are gatherings of stars destined to be connected together through their rays of opposite polarities. We define the unification graph of a constellation which is a multigraph specifying which stars can be connected together through the unification of their rays.

\begin{definition}[Duality]
\label{def:duality}
Two rays $r, r'$ are dual w.r.t. a set of colours $\mathcal{A}\subseteq\mathcal{C}$, written $r \poll_{\mathcal{A}} r'$, if there exists $c\in\mathcal{A}$ such that either $r=+c.t$ and $r'=-c.t'$, or $r=-c.t$ and $r'=+c.t'$, where $t,t'$ are matchable.
\end{definition}

\begin{definition}[Unification graph]
\label{def:ugraph}
The unification graph $\ugraph{\mathcal{A}}{\Sigma}$ of a constellation $\Sigma$ w.r.t. a set of colours $\mathcal{A}\subseteq\mathcal{C}$ is the undirected multigraph $(V^{\ugraph{\mathcal{A}}{\Sigma}},E^{\ugraph{\mathcal{A}}{\Sigma}},s^{\ugraph{\mathcal{A}}{\Sigma}})$ where $V^{\ugraph{\mathcal{A}}{\Sigma}}=\abs{\Sigma}$, $E^{\ugraph{\mathcal{A}}{\Sigma}}$ is the set of unordered pairs $\{ ((s,r),(s',r')) \in\rays{\Sigma} \mid  r \poll_{\mathcal{A}} r' \}$, and $s^{\ugraph{\mathcal{A}}{\Sigma}}$ is the function defined by $s^{\ugraph{\mathcal{A}}{\Sigma}}((s,r),(s',r'))=\{s,s'\}$.
\end{definition}

\begin{example}
\label{ex:ugraph}
If we take the constellations from example \ref{ex:constellation}, the unification graph $\ugraph{nat}{\Sigma_\nat}$ is an infinite linear graph and $\ugraph{b, c, d, e}{\Sigma_{ex}}$ is a linear graph
\begin{tikzpicture}
    \node[dot] at (0, 0) (a) {};
    \node[dot] at (0.5, 0) (b) {};
    \node[dot] at (1, 0) (c) {};
    \node[dot] at (1.5, 0) (d) {};
    \draw (a) -- (b);
    \draw (b) -- (c);
    \draw (c) -- (d);
\end{tikzpicture}
of $4$ vertices. The graph $\ugraph{a, b}{\Sigma_a}$ is a vertex with a loop.
\end{example}

The unification graph is used to define \emph{diagrams}, which are graphs obtained by composing occurrences of stars in a constellation along dual rays. They represent actual connexions between the stars of a constellation, following the connexions allowed by the unification graph.

\begin{definition}[Diagram]
\label{def:diagram}
Let $\mathcal{A}$ be a set of colours $\mathcal{A}\subseteq\mathcal{C}$. An $\mathcal{A}$-diagram $\delta$ over a constellation $\Sigma$ is a finite connected graph\footnote{The underlying graph is considered up to renaming of the vertices and edges. For practical purpose and without loss of generality, we will consider that vertices are natural numbers.} $D_{\delta}$ and a graph homomorphism $\delta : D_{\delta} \rightarrow \ugraph{\mathcal{A}}{\Sigma}$ which is \emph{ray injective}, i.e. for all $v\in V^{D_\delta}$ the function $\{ e\in E^{D_\delta}\mid v\in s^{D_\delta}(e)\}\rightarrow \abs{\delta(v)}; \{(\delta(v),r),(s',r')\}\mapsto r$ is injective.
If $D_{\delta}$ is a tree, we say $\delta$ is \emph{tree-like}.
\end{definition}

\begin{notation}
Given an $\mathcal{A}$-diagram $\delta$, we define its set of \emph{paired rays} as the set $\mathtt{paired}(\delta)=\{ (v,r) \mid \exists e\in E^{D_{\delta}}, v\in s^{D_{\delta}}(e), (\delta(v),r)\in\rays{\Sigma}\}$ and its set of free rays as $\mathtt{free}(\delta)=\{ (v,r) \mid v\in V^{D_{\delta}}, (\delta(v),r)\in\rays{\Sigma}, (v,r)\not\in \mathtt{paired}(\delta) \}$.
\end{notation}

\begin{example}
\label{ex:diagrams}
If we take the constellations from example \ref{ex:constellation} together with their unfication graph in example \ref{ex:ugraph}, any subchain of $\ugraph{nat}{\Sigma_\nat}$ and $\ugraph{b, c, d, e}{\Sigma_{ex}}$ can give rise to a diagram. As for, $\ugraph{a, b}{\Sigma_a}$, we can make a diagram be connecting as many occurrences of the star $\gstar{-a(x), +a(x), -b(x)}$ as we want.
\end{example}

\begin{remark}
The notion of diagram suggested by Girard \cite{girard2017transcendental} coincides with the notion of \emph{tree-like diagram} in this work. However, as we will explain below, the more general notion of diagram considered here allows for the representation of interesting models of computation, such as Wang tilings \cite{wang1961proving}.
\end{remark}

As mentioned above, we consider the underlying graphs of diagrams to have natural numbers as vertices. This is used here to produce a family of substitutions $(\theta(v,\_))_{v\in\nat}$ of pairwise disjoint codomains. This is used in the following definition to rename all occurrences of stars in diagrams with pairwise disjoint sets of free variables.

\begin{definition}[Underlying unification problem]
\label{def:underproblem}
The underlying unification problem $\mathcal{P}(\delta)$ of an $\mathcal{A}$-diagram $\delta$ over a constellation $\Sigma$ is defined as
\[ \{\theta(v,\floor{r}) \doteq \theta(v',\floor{r'}) \mid f\in E^{D_\delta}, s^{D_\delta}(f)=\{v,v'\}\wedge \delta(f)=((\delta(v),r),(\delta(v'),r')) \}.\]
\end{definition}

In some ways, diagrams generalise trails in a graph. Where paths describe a possible trajectory for a particle, diagrams describe the possible trajectories of a wave that can simultaneously spread in several directions when encountering forks. We now introduce the notion of \emph{saturated diagrams}, which correspond to \emph{maximal paths}.

\begin{definition}[Saturated Diagram]
\label{def:diagramorder}
We define a preorder $\sqsubseteq$ on $\mathcal{A}$-diagrams over a constellation $\Sigma$ by: $\delta \sqsubseteq \delta'$ if there exists an isomorphism $\varphi$ from a graph $D$ of $D_{\delta'}$ to $D_{\delta}$ such that $\delta=\delta'\circ\varphi$.

A maximal $\mathcal{A}$-diagram w.r.t. $\sqsubseteq$ is called \emph{saturated}.
\end{definition}

\begin{example}
For the constellations from example \ref{ex:constellation} and their unfication graph in example \ref{ex:ugraph}, $\Sigma_\nat$ has no saturated diagram since it is always possible to extend any star $\gstar{-nat.s^n(0), +nat.s^{n+1}(0)}$ by connecting it to
\[\gstar{-nat.s^{n+1}(0), +nat.s^{n+2}(0)} \in \Sigma_\nat.\]
Similarly, $\Sigma_a$ can always connect any star with a new copy of itself so no saturated diagram can be constructed. The star $\gstar{x, -b(g(x))} \in \Sigma_{ex}$ can be connected to either $\gstar{g(x), f(x), +a(f(x))}+\gstar{-a(y), +b(y)}$ or just $\gstar{+b(x), x}$ and these two connexions can't be extended. Therefore, they form saturated diagrams.
\end{example}

\subsection{Evaluation and Normal form}\label{subsec:evaluation}

Diagrams can be evaluated by confronting and annihilating the rays connected together similarly to a chemical reaction between molecules which will have an effect on the remaining free rays. This reaction defines a unification problem. When a solution (a substitution) to this problem exists, one can apply it to all the free rays involved in the diagram. Edges in the diagram then represent equalities and the diagram can be \emph{contracted} into a single star. This operation is called \emph{actualisation}.

\begin{definition}[Correct diagram and actualisation]
\label{def:correctdiagram}
Consider a constellation $\Sigma$ and a set of colours $\mathcal{A}\subseteq\mathcal{C}$. An $\mathcal{A}$-diagram $\delta$ is correct if $\mathtt{free}(\delta)\neq\emptyset$ and the associated unification problem $\mathcal{P}(\delta)$ has a solution.

The actualisation of a correct diagram $\delta$ is the star $\actu \delta$ defined as $\abs{\actu \delta}=\mathtt{free}(\delta)$ and $\rho_{\actu \delta}: (e,r)\in\mathtt{free}(\delta)\mapsto \psi(\theta(e,r))$, where $\psi = \unifalgo{\mathcal{P}(\delta)}$ and $\theta(e,\_)$ is the renaming used in Definition \ref{def:underproblem}.
\end{definition}

\begin{notation}
We write $\correctdiag{k}(\Sigma)$ (resp. $\correcttreediag{k}(\Sigma)$) for the set of correct $\mathcal{A}$-diagrams (resp. correct tree-like $\mathcal{A}$-diagrams) over $\Sigma$ with $k$ vertices, and $\saturateddiag{k}(\Sigma)$ (resp. $\saturatedtreediag{k}(\Sigma)$) the set of saturated $\mathcal{A}$-diagrams in $\correctdiag{k}(\Sigma)$ (resp. in $\correcttreediag{k}(\Sigma)$).
\end{notation}

Note that the approach proposed by Girard used a step-by-step procedure merging two adjacent stars in a diagram. It is possible to prove that his approach and our lead to the same result, and more precisely that a sequence of merging coincides with a step-by-step execution of the Martelli-Montanari unification algorithm \cite{martelli1982efficient}. Stating and proving this refined result is however outside of the scope of this paper.

\begin{definition}[Normalisation]
\label{def:normalalization}
The normalisation or execution (resp. tree-like normalisation) of a constellation $\Sigma$ w.r.t. a set of colours $\mathcal{A}\subseteq\mathcal{C}$ is defined by
$\exec^*_\mathcal{A}(\Sigma) = \bigcup_{k=0}^{\infty} \actu\saturatedstardiag{k}(\Sigma)$ ($*\in\{\emptyset,tree\}$), where $\actu\saturatedstardiag{k}(\delta)$ denotes the set $\{\actu\delta\mid \delta\in\saturatedstardiag{k}(\delta)\}$.
\end{definition}

\begin{example}
The constellation $\Sigma_{ex}$ is the only constellation from example \ref{ex:constellation} which has saturated diagrams. The only correct diagram $\delta$ connects $\gstar{x, -b(g(x))}$ and $\gstar{+b(y), y}$. We have to solve the problem $\mathcal{P}(\delta) = \{b(g(x)) \doteq b(y)\}$ producing $\theta = \{y \mapsto g(x)\}$ and we finally obtain $\exec(\Sigma_{ex}) = \theta\gstar{x, y} = \gstar{x, g(x)}$ since $\Sigma_{ex}$ has only two diagrams. One can check that the other diagram fails.
\end{example}

\begin{definition}[Strong normalisation]
\label{def:strongnormal}
A constellation $\Sigma$ is strongly normalising w.r.t. a set of colours $\mathcal{A}\subseteq\mathcal{C}$ if and only if $\exec(\Sigma)$ is a finite constellation. When $\mathcal{A}=\mathcal{C}$, we simply say that $\Sigma$ is strongly normalising.
\end{definition}

\begin{theorem}[Associativity / Church-Rosser]
\label{thm:associativity}
Let $\Sigma$ be a constellation, and $\mathcal{A,B}$ be two disjoint sets of colours, i.e. $\mathcal{A}\subseteq\mathcal{C}$, $\mathcal{B}\subseteq\mathcal{C}$, and $\mathcal{A}\cap\mathcal{B}=\emptyset$. Then:
\[ \exec^*_\mathcal{B}(\exec^*_\mathcal{A}(\Sigma)) = \exec^*_\mathcal{A\cup B}(\Sigma) = \exec^*_\mathcal{A}(\exec^*_\mathcal{B}(\Sigma)),\quad \mathrm{(*\in\{\emptyset,tree\}}) \]
\end{theorem}

\begin{proof}[Proof (Sketch).]
We establish the first isomorphism, which is enough by symmetry.
First note that the disjointness of the sets of colours implies that the set of edges in $\ugraph{\mathcal{A}\cup\mathcal{B}}{\Sigma}$ is the disjoint union of the sets of edges in $\ugraph{\mathcal{A}}{\Sigma}$ and those in $\ugraph{\mathcal{B}}{\Sigma}$. As a consequence, any diagram on $\ugraph{\mathcal{A}}{\Sigma}$ (resp. $\ugraph{\mathcal{B}}{\Sigma}$) can be thought of as a diagram on $\ugraph{\mathcal{A}\cup\mathcal{B}}{\Sigma}$.

Let $\delta : D_\delta \rightarrow \ugraph{\mathcal{B}}{\exec^*_\mathcal{A}(\Sigma)}$ be a correct saturated $\mathcal{B}$-diagram of $\exec^*_\mathcal{A}(\Sigma)$. For each $s\in\abs{\exec^*_\mathcal{A}(\Sigma)}$, the star $\exec^*_\mathcal{A}(\Sigma)(s)$, which we will write $\sigma_s$, corresponds to a diagram $\delta_s: D_s\rightarrow \ugraph{\mathcal{A}}{\Sigma}$.
We can therefore build a diagram $\bar{\delta}$ over $\ugraph{\mathcal{A}\cup\mathcal{B}}{\Sigma}$ by \emph{blowing up} $\delta$ along the diagrams $\delta_s$. More precisely, we construct the graph $\bar{D}_\delta$ obtained by replacing each vertex $s$ by the graph $D_s$; this is well-defined as each ray in $\delta(s)$ comes from a unique ray from a star $\delta_s(s')$ for some $s'\in V^{D_s}$, and therefore each edge $e$ in $D_\delta$ of source $s$ becomes an edge of source the unique star $s'\in V^{D_s}$. The morphism $\delta$ then extends uniquely to a morphism $\bar{\delta}: \bar{D}_\delta\rightarrow \ugraph{\mathcal{A}\cup\mathcal{B}}{\Sigma}$ whose action on the subgraphs $D_s$ coincides with that of $\delta_s$ (as a morphism into $\ugraph{\mathcal{A}\cup\mathcal{B}}{\Sigma}$).

To check that this mapping from $\mathcal{B}$-diagrams on $\exec^*_\mathcal{A}(\Sigma)$ to $(\mathcal{A\cup B})$-diagrams on $\Sigma$ is indeed an isomorphism, one can directly define an inverse mapping. For this purpose, the essential remark is that given a $(\mathcal{A\cup B})$-diagram $\bar{\delta}$ on $\Sigma$, one can recover the underlying $\mathcal{B}$-diagrams on $\exec^*_\mathcal{A}(\Sigma)$ as the restriction of $\bar{\delta}$ to the connected components of the graph obtained from $D_{\bar{\delta}}$ by removing the edges mapped to $\mathcal{A}$-coloured edges in $\ugraph{\mathcal{A}\cup\mathcal{B}}{\Sigma}$. The underlying graph of the corresponding $\mathcal{B}$-diagram on $\exec^*_\mathcal{A}(\Sigma)$ is then defined from $D_{\bar{\delta}}$ by contracting each of these connected components to a single vertex.
\end{proof}

\begin{remark}
In Girard's first paper on Transcendental Syntax \cite{girard2017transcendental}, the constellation $\Sigma = \gstar{+a.x, -a.x, +b.x}$ is mentioned as a counter-example for the Church-Rosser property. Here, we have $\exec^{\mathrm{tree}}_{\{a\}}(\exec^{\mathrm{tree}}_{\{b\}}(\Sigma)) = \exec^{\mathrm{tree}}_{\{b\}}(\exec^{\mathrm{tree}}_{\{a\}}(\Sigma)) = \emptyset$ (because no saturated diagram on $a$ nor on $b$ can be constructed). Our understanding of Girard's failure comes from his limitation to strongly normalising constellations, so that $\exec^{\mathrm{tree}}_{\{a\}}(\Sigma)$ was not defined.
\end{remark}

\subsection{Remarkable examples of constellations}\label{subsec:examples}

We now explain how the stellar resolution model generalises several established frameworks from the literature.

\subsubsection{(Hyper)graphs} As a first example, we detail how graphs (and more generally hypergraphs) can be encoded, and how the execution coincides with the computation of paths. A consequence of this result is that models of computation such as \emph{Interaction Graphs} \cite{seiller2012interaction,seiller2016interaction} are special cases of our construction.

\begin{definition}[Encoding of hypergraphs]
\label{ex:graphencoding}
Let $G=(V,E,s,t)$ be a directed hypergraph, i.e. $s,t$ are maps $E\rightarrow\wp(V)$. We suppose given a signature in which there exists distinct unary function symbols to represent the elements of $V$ and $E$. Then, for each $e\in E$ with $s(e)=\{v_1,\dots, v_n\}$ and $t(e)=\{w_1,\dots,w_m\}$, we define the star
\[ \bar{e}=\gstar{-v_1(x), -v_2(x),\dots -v_n(x), +w_1(e(x)), +w_2(e(x)),\dots, +w_m(e(x))}.\]
The graph $G$ is then represented by the constellation $\bar{G}$ defined as the multiset $\alpha: E\rightarrow \bar{e}$.
\end{definition}

\noindent One can check that this encoding has the property that diagrams $\rho: T\rightarrow U(\bar{G})$ in which $T$ is a \emph{linear graph} (i.e. a graph $\bullet \rightarrow \bullet\rightarrow \dots\rightarrow\bullet$) are in bijection with walks in $G$; moreover the diagram is saturated if and only if the walk is \emph{maximal}.

\subsubsection{Logic programming}

The stellar resolution is very close to concepts coming from logic programming since they all use first-order terms to do computations. The main difference is that logic programming has practical motivations and enjoy logical interpretations through model theory while our approach is purely computational and doesn't obey logic. We choose to encode first-order disjunctive clauses \cite{robinson1965machine}.

\begin{definition}[Encoding of first-order disjunctive clauses]
We have the following encoding for first-order atoms: $P(t_1, ..., t_n)^\bigstar = +P(t_1, ..., t_n)$, as well as $(\lnot P(t_1, ..., t_n))^\bigstar = -P(t_1, ..., t_n)$ where $P$ is a colour. A clause is encoded as $\{A_1, ..., A_n\}^\bigstar = \gstar{A_1^\bigstar, ..., A_n^\bigstar}$ where $A_1, ..., An$ are first-order atoms. A set of clauses is translated into $\{C_1, ..., C_n\}^\bigstar = C_1^\bigstar + ... + C_n^\bigstar$. Clauses are connected through the resolution rule: from $\Gamma \cup \{P(t_1, ..., t_n)\}$ and $\Delta \cup \{\lnot P(u_1, ..., u_n)\}$, we infer $\theta(\Gamma \cup \Delta)$ where $\theta = \unifalgo{\{t_i \doteq u_i\}_{1 \leq i \leq n}}$.
\end{definition}

The actualisation of a diagram corresponds to several applications of the resolution rule. In particular, a saturated diagram is a full computation (instead of a partial/unfinished computation). The unification graph we use is reminiscent of dependency graphs \cite{nguyen2007termination, dimopoulos1996graph} and the normalisation of the computation of answer sets \cite{gelfond2008answer, eiter2009answer}. Although not described here, it may also be possible to simulate goal-directed inferences (such as SLD \cite{kowalski1974predicate} or SLO \cite{lobo1991semantics} resolution) by using unpolarised rays.

Another close model is the model of flows \cite{girard1995geometry, bagnol2014resolution} made of couples of terms $t \leftharpoonup u$ where $t$ and $u$ have exactly the same variables.

\begin{definition}[Encoding of flows]
Let $f = t \leftharpoonup u$ be a flow. We define its encoding as $f^\bigstar = \gstar{+t, -u}$. A wiring $F = f_1+...+f_n$ is a set of flows and it is encoded into $F^\bigstar = f_1^\bigstar + ... + f_n^\bigstar$.
\end{definition}

A wiring will describe a graph where flows represent edges. Matching (definition \ref{def:matching}) decides which vertices the edges connect. The execution of a wiring $\exec(F)$ computing maximal paths in corresponding graph corresponds to $\exec(F^\bigstar)$.

\subsubsection{Wang tiles and abstract tile assembly models}

Wang tiles \cite{wang1961proving} are square tiles with a colour on each side. For instance:
$W = \left\{\scalebox{0.3}{
    \wang{yellow}{red}{blue}{yellow}
    \wang{green}{green}{blue}{red}
}\right\}$. We are interested in the possible tilings made by putting copies of the tiles side by side such that two adjacent sides have the same colour. Our setting naturally generalises the structure and behaviour of Wang tiles in $\nat^2$.

\begin{definition}[Encoding of Wang tiles]
Let $t^i = (c^i_w, c^i_e, c^i_s, c^i_n)$, $i=1,\dots,k$, be a set of Wang tile characterised by west, east, south and north colours. We encode each Wang tile in $\nat^2$ by the star
\[ (t^i)^\bigstar = \gstar{-h(c^i_w(x), x, y), -v(c^i_s(y), x, y), +h(c^i_e(s(x)), x, y), +v(c^i_n(s(y)), x, y)}.\]
\end{definition}

This supposes the signature to contain a unary function symbol for each colour used in the Wang tiles set considered. We use two additional ternary colours $h$ (for \emph{horizontal axis}) and $v$ (for \emph{vertical axis}) for the rays. The terms $s^n(x), s^n(y)$ for $n \leq 0$ represent coordinates on the horizontal and vertical axis encoded as natural numbers, and forbids geometric patterns that would be impossible in the original Wang tile framework (e.g. a tile connected to itself).
A (possibly cyclic) diagram will correspond to a finite tiling.

Wang tilings are a Turing-complete model of computation \cite{berger1966undecidability}. As the simulation works by representing the space-time diagram of the execution, it is enough to restrict to Wang tiles over $\nat^2$ by considering machines with a single right-infinite tape (i.e. Turing's original \emph{a-machine} model \cite{turing1936computable}). As a consequence, we obtain the following theorem.

\begin{theorem}[Turing-completeness]
Stellar resolution can simulate Turing machines.
\end{theorem}

Wang tiles are very close to the abstract tile assembly models (aTAM) \cite{winfree1998algorithmic, patitz2014introduction} which can be used as a model of computation with applications in DNA computing \cite{seeman1982nucleic}. Our model is able to represent aTAM, and even generalise it.

A tile $t = (g_w, g_e, g_s, g_n)$ in the aTAM is a Wang tile where the sides $g \in \mathcal{G}$ (called glues types instead of colours) are associated to an integer $\mathrm{str}(g) \geq 0$ called their strength. An additional model parameter is called the \emph{temperature}, which we write $\mathtt{t}$. A tile can be attached to an assembly (i.e. a diagram in the stellar resolution terminology) if adjacent sides have matching glue type and strength, and the total strength of its connections is at least\footnote{As a consequence, a tile can connect to an assembly in temperature 2 through two faces of strength 1, but cannot connect through a single of those faces.} $\mathtt{t}$. To encode those we use colours $h,v$ of arity 3 as above but decline them in pairs $h_{-},h_{+}$ and $v_{-},v_{+}$, together with a binary symbol $G$, and define for all glue type $g \in \mathcal{G}$ the term $\mathrm{gl}(g)(x)=G(g(x), s^{\mathrm{str}(g)}(0))$. Temperature is dealt with by adding \emph{ambiance stars} that are required to connect between tile encodings. Ambiance stars connect to a tile translation on several sides and force the sum of forces to be larger than the temperature. Ambiance stars $A^{a,b,c,d}$ are defined as
\[
\begin{array}{p{12cm}}
$[-v_{-}(G(x_a, s^{a}(y_a)), z, w),$\\
\hspace{1cm} $+v_{+}(G(x_b, s^{b}(y_b)), z, w),$\\
\hspace{2cm} $-h_{-}(G(x_c, s^{c}(y_c)), z, w),$\\
\hspace{3cm} $+h_{+}(G(x_d, s^{d}(y_d)), z, w),$\\
\hspace{4cm} $+v_{+}(G(x_a,s^{a}(y_a),z_1,w_1)),$\\
\hspace{5cm} $-v_{-}(G(x_b,s^{b}(y_b),z_2,w_2)),$\\
\hspace{6cm} $+h_{+}(G(x_c, s^{c}(y_c)), z_3, w_3),$\\
\hspace{7cm} $-h_{-}(G(x_d, s^{d}(y_d)), z_4, w_4)]$
\end{array}
\]
for all $0\leqslant a,b,c,d\leqslant \mathtt{t}$ such that $a+b+c+d=\mathtt{t}$. We also require the addition of stars $[-v_{+}(x)]$, $[+v_{-}(x)]$, $[-h_{+}(x)]$, $[+v_{-}(x)]$ to \emph{plug} connexions where no tiles can be added.

\begin{definition}[Encoding of self-assembling tiles]
Let $t^i = (g^i_w, g^i_e, g^i_s, g^i_n)$, $i\in I$ (countable but possibly infinite) be a set of self-assembling tiles. The tile $t^i$ is encoded as the star
$(t^i)^\bigstar$:
\[\begin{array}{c}
[-h_{+}(\mathrm{gl}(g^i_w)(x), x, y),\\
\hspace{2cm} -v_{+}(\mathrm{gl}(g^i_s)(y), x, y),\\
\hspace{4cm} +h_{-}(\mathrm{gl}(g^i_e)(s(x)), x, y),\\
\hspace{6cm} +v_{-}(\mathrm{gl}(g^i_n)(s(y)), x, y)].\end{array}\]

The aTAM model with tiles $(t^i)_{i\in I}$ and temperature $\mathtt{t}$ is encoded as \[\sum_{a+b+c+d=\mathtt{t}}A^\mathtt{a,b,c,d}+\sum_{i\in I} (t^i)^\bigstar+[-v_{+}]+[+v_{-}]+[-h_{+}]+[+v_{-}].\]
\end{definition}

\section{Interpreting the computational content of MLL}\label{sec:computational}

In this section, we restrict the definitions to the tree-like normalisation i.e $\exec(\Sigma)$ will be a shortcut for $\exec_{\mathcal{A}}^{\mathrm{tree}}(\Sigma)$ where $\mathcal{A}$ is a set of colour left implicit.

\subsection{Multiplicative Linear Logic}\label{subsec:mll}

\begin{figure}[t]
    \begin{subfigure}{.275\textwidth}
        \[A, B = X_i \mid X_i^\bot \mid A \otimes B \mid A \parr B \quad \text{where $i \in \nat$}\]
        \caption{{\footnotesize MLL Formulas}}
        \label{subfig:mllformulas}
    \end{subfigure}

    \vspace{1em}
    \begin{subfigure}{.275\textwidth}
        \scalebox{0.9}{
        \begin{prooftree}
            \infer0{ \vdash A, A^\bot }
        \end{prooftree}
        \qquad
        \begin{prooftree}
            \hypo{ \vdash \Gamma, A \quad \vdash \Delta, B }
            \infer1{ \vdash \Gamma, \Delta, A \otimes B }
        \end{prooftree}}

        \medskip
        \scalebox{0.9}{
        \begin{prooftree}
            \hypo{\vdash \Gamma, A \quad \vdash \Delta, A^\bot}
            \infer1{\vdash \Gamma, \Delta}
        \end{prooftree}
        \qquad
        \begin{prooftree}
            \hypo{\vdash \Gamma, A, B}
            \infer1{\vdash \Gamma, A \parr B}
        \end{prooftree}}
        \caption{Sequent calculus rules}
        \label{subfig:mllsequent}
    \end{subfigure}
    \hspace{2cm}
    \begin{subfigure}{.4\textwidth}
        \scalebox{0.75}{
        \begin{minipage}{9cm}
        \begin{tikzpicture}
          \node[dot] at (0, 0) (a) {};
          \node[dot] at (2, 0) (ad) {};
          \node at (1, 0.75) (ax) {ax};
          \draw[-latex, rounded corners=5pt] (ax) -| (a);
          \draw[-latex, rounded corners=5pt] (ax) -| (ad);
          \node at (1, -1) (label) {Axiom};
        \end{tikzpicture}
        \quad
        \begin{tikzpicture}
          \node[dot] at (0, 0) (a) {};
          \node[dot] at (2, 0) (ad) {};
          \node at (1, -0.75) (cut) {cut};
          \draw[-latex, rounded corners=5pt] (a) |- (cut);
          \draw[-latex, rounded corners=5pt] (ad) |- (cut);
          \node at (1, -1.5) (label) {Cut};
        \end{tikzpicture}
        \quad
        \begin{tikzpicture}
          \node[dot] at (-0.75, 0.75) (a) {};
          \node[dot] at (0.75, 0.75) (b) {};
          \node at (0, 0) (tens) {$\otimes$};
          \node[dot] at (0, -0.75) (ab) {};
          \draw[-stealth] (a) -- (tens);
          \draw[-stealth] (b) -- (tens);
          \draw[-stealth] (tens) -- (ab);
          \node at (0, -1.5) (label) {Tensor};
        \end{tikzpicture}
        \quad
        \begin{tikzpicture}
          \node[dot] at (-0.75, 0.75) (a) {};
          \node[dot] at (0.75, 0.75) (b) {};
          \node at (0, 0) (par) {$\parr$};
          \node[dot] at (0, -0.75) (ab) {};
          \draw[-stealth] (a) -- (par);
          \draw[-stealth] (b) -- (par);
          \draw[-stealth] (par) -- (ab);
          \node at (0, -1.5) (label) {Par};
        \end{tikzpicture}
        \end{minipage}}
        \caption{{\footnotesize Proof-structures}}
        \label{subfig:proofstructures}
    \end{subfigure}

    \begin{subfigure}{0.9\textwidth}
        \scalebox{.75}{
          \begin{minipage}{0.25\textwidth}
            \begin{tikzpicture}
                \node[dot, label={-90:$A$}] at (0, 0) (a) {};
                \node[dot, label={0:$A^\bot$}] at (1, 0) (ad) {};
                \node at (0.5, 0.75) (ax) {ax};
                \node at (1.75, -0.75) (cut) {cut};
                \node[dot, label={90:$A$}] at (2.5, 0) (a2) {};
                \draw[-latex, rounded corners=5pt] (ax) -| (a);
                \draw[-latex, rounded corners=5pt] (ax) -| (ad);
                \draw[-latex, rounded corners=5pt] (ad) |- (cut);
                \draw[-latex, rounded corners=5pt] (a2) |- (cut);
            \end{tikzpicture}
          \end{minipage}
          $\quad\leadsto\quad$
          \begin{minipage}{0.1\textwidth}
            \begin{tikzpicture}
             \node at (0, 0) (a) {$A$};
           \end{tikzpicture}
          \end{minipage}
          \hspace{0.5cm}
          \begin{minipage}{0.4\textwidth}
            \begin{tikzpicture}
                \node[dot, label={-180:$A$}] at (-0.75, 0.75) (a) {};
                \node[dot, label={-180:$B$}] at (0.75, 0.75) (b) {};
                \node at (0, 0) (tens) {$\otimes$};
                \node[dot, label={-180:$A \otimes B$}] at (0, -0.75) (ab) {};
                \draw[-stealth] (a) -- (tens);
                \draw[-stealth] (b) -- (tens);
                \draw[-stealth] (tens) -- (ab);

                \node[dot, label={-180:$A^\bot$}] at (2, 0.75) (c) {};
                \node[dot, label={-180:$B^\bot$}] at (3.5, 0.75) (d) {};
                \node at (2.75, 0) (par) {$\parr$};
                \node[dot, label={-180:$A^\bot \parr B^\bot$}] at (2.75, -0.75) (cd) {};
                \draw[-stealth] (c) -- (par);
                \draw[-stealth] (d) -- (par);
                \draw[-stealth] (par) -- (cd);

                \node at (1.25, -1.25) (cut) {cut};
                \draw[-latex, rounded corners=5pt] (ab) |- (cut);
                \draw[-latex, rounded corners=5pt] (cd) |- (cut);
            \end{tikzpicture}
          \end{minipage}
          $\quad\leadsto\quad$
          \begin{minipage}{0.3\textwidth}
            \begin{tikzpicture}
              \node[dot, label={-180:$A$}] at (0, 0) (a) {};
              \node[dot, label={-180:$B$}] at (1.5, 0) (b) {};
              \node[dot, label={-180:$A^\bot$}] at (3, 0) (c) {};
              \node[dot, label={-180:$B^\bot$}] at (4.5, 0) (d) {};
              \node at (0.75, -0.5) (cut1) {cut};
              \node at (3, -1) (cut2) {cut};
              \draw[-latex, rounded corners=5pt] (a) |- (cut1);
              \draw[-latex, rounded corners=5pt] (c) |- (cut1);
              \draw[-latex, rounded corners=5pt] (b) |- (cut2);
              \draw[-latex, rounded corners=5pt] (d) |- (cut2);
           \end{tikzpicture}
          \end{minipage}} 
        \caption{{\footnotesize Cut-elimination reductions: axiomatic cut (left), multiplicative cut (right)}}
        \label{subfig:cutelimination}
    \end{subfigure}
\caption{Syntax of Multiplicative Linear Logic (MLL)}
\label{fig:mll}
\end{figure}

Multiplicative linear logic (MLL) is a fragment of linear logic \cite{girard1987linear} restricted to the tensor $\otimes$ and par $\parr$ connectives. The set $\mathcal{F}_{\mathrm{MLL}}$ of MLL formulas is defined by the grammar of the figure \ref{subfig:mllformulas}. Linear negation $(\cdot)^\bot$ is extended to formulas by involution and de Morgan laws: $X^{\bot\bot} = X$, $(A \otimes B)^\bot = A^\bot \otimes B^\bot$, and $(A \parr B)^\bot = A^\bot \parr B^\bot$.

Proofs of MLL can be written in the traditional sequent calculus fashion \cite{gentzen1935untersuchungen, gentzen1935untersuchungen2} using the set of rules shown in figure \ref{subfig:mllsequent}. In his seminal paper \cite{girard1987linear}, Girard introduced an alternative syntax, akin to \emph{natural deduction}, based on a graph-theoretic representation of proofs. In this syntax, one considers the notion of proof-structure, i.e. an directed hypergraph with vertices labelled by formulas and constructed from hyperedges\footnote{For practical purposes, the source edges are ordered, and we will talk about the "left" and "right" sources since there never are more than two; illustrations implicitly represent the left (resp. right) source on the left (resp. right).} labelled within $\{ax, cut, \otimes, \parr\}$ and satisfying the arities and labelling constraints shown in figure \ref{subfig:proofstructures}. A proof-structure also satisfies the additional constraint that each vertex must be (1) the target of exactly one hyperedge, and (2) the source of at most one hyperedge. When needed, a proof-structure will be defined as a 6-tuple $(V,E,s,t,\ell_V,\ell_E)$, where $(V,E,s,t)$ is a directed hypergraph (see Definition \ref{ex:graphencoding}) and $\ell_V: V\rightarrow \mathcal{F}_{\mathrm{MLL}}$, $\ell_E: E\rightarrow\{\otimes,\parr,\mathrm{ax},\mathrm{cut}\}$ are labelling maps.

The cut-elimination procedure, which is defined in the natural way for MLL sequent calculus, becomes a graph-rewriting system on proof-structures, defined by the two rewriting rules (figure \ref{subfig:cutelimination}). The following definition explains how sequent calculus proofs can be represented as proof-structures.

\begin{definition}[Translation of MLL sequent calculus]
\label{def:translation}
We define a translation $\interp{\cdot}$ from MLL sequent calculus derivations to proof-structures. The proof-structures of the domain of $\interp{\cdot}$ are called the \emph{proof-nets}.

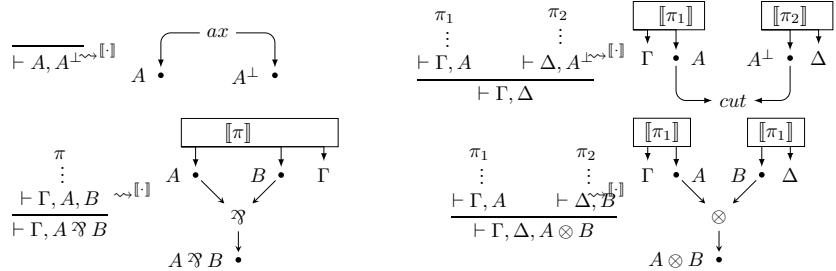
\scalebox{.75}{
\begin{minipage}{0.1\textwidth}
    \begin{prooftree}
        \infer0{\vdash A, A^\bot}
    \end{prooftree}
\end{minipage}
$\leadsto^{\interp{\cdot}}$
\begin{minipage}{0.4\textwidth}
    \begin{tikzpicture}
      \node[dot, label={-180:$A$}] at (0, 0) (a) {};
      \node[dot, label={-180:$A^\bot$}] at (2, 0) (ad) {};
      \node at (1, 0.75) (ax) {$ax$};
      \draw[-latex, rounded corners=5pt] (ax) -| (a);
      \draw[-latex, rounded corners=5pt] (ax) -| (ad);
    \end{tikzpicture}
\end{minipage}

\begin{minipage}{0.25\textwidth}
    \begin{prooftree}
        \hypo{\pi_1}
        \ellipsis{}{\vdash \Gamma, A}
        \hypo{\pi_2}
        \ellipsis{}{\vdash \Delta, A^\bot}
        \infer2{\vdash \Gamma, \Delta}
    \end{prooftree}
\end{minipage}
$\leadsto^{\interp{\cdot}}$
\begin{minipage}{0.6\textwidth}
    \begin{tikzpicture}
      \draw (-0.75,0.5) rectangle (0.5,1);
      \node at (0, 0.75) (label) {$\interp{\pi_1}$};
      \draw (1.5,0.5) rectangle (2.75,1);
      \node at (2, 0.75) (label) {$\interp{\pi_2}$};
      \node at (-0.5, 0) (gamma) {$\Gamma$};
      \node at (2.5, 0) (delta) {$\Delta$};
      \node[dot, label={0:$A$}] at (0, 0) (a) {};
      \node[dot, label={180:$A^\bot$}] at (2, 0) (ad) {};
      \draw[-latex] (-0.75, 0.5) -| (a);
      \draw[-latex] (-0.75, 0.5) -| (gamma);
      \draw[-latex] (1.5, 0.5) -| (ad);
      \draw[-latex] (1.5, 0.5) -| (delta);
      \node at (1, -0.75) (cut) {$cut$};
      \draw[-latex, rounded corners=5pt] (a) |- (cut);
      \draw[-latex, rounded corners=5pt] (ad) |- (cut);
    \end{tikzpicture}
\end{minipage}
}

\scalebox{.75}{
\begin{minipage}{0.15\textwidth}
    \begin{prooftree}
        \hypo{\pi}
        \ellipsis{}{\vdash \Gamma, A, B}
        \infer1{\vdash \Gamma, A \parr B}
    \end{prooftree}
\end{minipage}
$\leadsto^{\interp{\cdot}}$
\begin{minipage}{0.4\textwidth}
    \begin{tikzpicture}
      \draw (-1,1.25) rectangle (1.75,1.75);
      \node at (0, 1.5) (label) {$\interp{\pi}$};
      \node at (1.5, 0.75) (gamma) {$\Gamma$};
      \node[dot, label={180:$A$}] at (-0.75, 0.75) (a) {};
      \node[dot, label={180:$B$}] at (0.75, 0.75) (b) {};
      \draw[-latex] (-1, 1.25) -| (gamma);
      \draw[-latex] (-1, 1.25) -| (a);
      \draw[-latex] (-1, 1.25) -| (b);
      \node at (0, 0) (par) {$\parr$};
      \node[dot, label={180:$A \parr B$}] at (0, -0.75) (ab) {};
      \draw[-stealth] (a) -- (par);
      \draw[-stealth] (b) -- (par);
      \draw[-stealth] (par) -- (ab);
    \end{tikzpicture}
\end{minipage}

\begin{minipage}{0.2\textwidth}
    \begin{prooftree}
        \hypo{\pi_1}
        \ellipsis{}{\vdash \Gamma, A}
        \hypo{\pi_2}
        \ellipsis{}{\vdash \Delta, B}
        \infer2{\vdash \Gamma, \Delta, A \otimes B}
    \end{prooftree}
\end{minipage}
$\leadsto^{\interp{\cdot}}$
\begin{minipage}{0.4\textwidth}
    \begin{tikzpicture}
      \draw (-1.5,1.25) rectangle (-0.5,1.75);
      \draw (0.5,1.25) rectangle (1.5,1.75);
      \node at (-1, 1.5) (label) {$\interp{\pi_1}$};
      \node at (1, 1.5) (label) {$\interp{\pi_1}$};
      \node at (-1.25, 0.75) (gamma) {$\Gamma$};
      \node at (1.25, 0.75) (delta) {$\Delta$};
      \node[dot, label={0:$A$}] at (-0.75, 0.75) (a) {};
      \node[dot, label={180:$B$}] at (0.75, 0.75) (b) {};
      \draw[-latex] (-1.5, 1.25) -| (gamma);
      \draw[-latex] (-1.5, 1.25) -| (a);
      \draw[-latex] (1.5, 1.25) -| (delta);
      \draw[-latex] (1.5, 1.25) -| (b);
      \node at (0, 0) (par) {$\otimes$};
      \node[dot, label={180:$A \otimes B$}] at (0, -0.75) (ab) {};
      \draw[-stealth] (a) -- (par);
      \draw[-stealth] (b) -- (par);
      \draw[-stealth] (par) -- (ab);
    \end{tikzpicture}
\end{minipage}
}
\end{definition}

Note that this translation is not surjective, and some proof-structures do not represent sequent calculus proofs. This is tackled by the \emph{correctness criterion} which characterises those proof-structures that do translate sequent calculus proofs through topological properties. This is discussed in the next section but for the time being we define the notion of proof-net.

\begin{definition}[Proof-nets]
A \emph{proof-net} is a proof-structure $\mathcal{S}$ such that there exists a MLL sequent calculus proof $\pi$ with $\mathcal{S}=\interp{\pi}$.
\end{definition}

\subsection{Reconstruction of proof-structures and their dynamics}\label{subsec:proofs}

Let us note that proof-structures can be defined inductively. A proof-structure with only one hyperedge is necessarily an axiom with two conclusions. Then a proof-structure with $n$ hyperedges is either built from the union of two proof-structures with respectively $k$ and $n-k$ hyperedges, or from a proof-structure with $n-1$ hyperedges extended by either a $\otimes$, $\parr$, or $cut$ hyperedge on two of its conclusions. In the following we use this inductive definition to define the \emph{address} of occurrences of atoms in a proof-structure.

\begin{definition}
A vertex $v$ is \emph{above} another vertex $u$ in a proof-structure if there exists a directed path from $v$ to $u$ going through only $\otimes$ and $\parr$ hyperedges.
\end{definition}

\begin{definition}[Address]
\label{def:loc}
We now consider a signature in which there exists at least two unary function symbols $\mathtt{r},\mathtt{l}$ and unary functions symbols $p_A$ for all occurrences\footnote{The set of formulas $\mathcal{F}$ is countable, and there are only finite numbers of occurrences of a given formula in a given proof-structure, hence the set $\mathcal{F}\times\nat$ suffices and is still countable.} of MLL formulas $A$. We define the partial address $\ploc_\mathcal{S}(d,x)$ of an occurrence of atom $d$ in a MLL proof-structure $\mathcal{S}$, with respect to the variable $x$, inductively:
\begin{itemize}
\item $\ploc_\mathcal{S}(d,x)=x$ when $\mathcal{S}$ consists only of an axiom hyperedge;
\item $\ploc_\mathcal{S}(d,x)=\ploc_{\mathcal{S}_i}(d,x)$ if $\mathcal{S}$ is the union of two smaller proof-structures $\mathcal{S}_1, \mathcal{S}_2$ and $d$ appears in $\mathcal{S}_i$;
\item $\ploc_\mathcal{S}(d,x)=\mathtt{l}(\ploc_{\mathcal{S}'}(d,x)$ (resp. $\ploc_\mathcal{S}(d,x)=\mathtt{r}(\ploc_{\mathcal{S}'}(d,x)$) if $\mathcal{S}$ is obtained from $\mathcal{S}'$ by adding a $\otimes$ or $\parr$ hyperedge $e$, and if $d$ is above the left source (resp. the right source) of $e$.
\item $\ploc_\mathcal{S}(d,x)=\ploc_{\mathcal{S}'}(d,x)$ otherwise.
\end{itemize}
The partial address of $d$ is defined with respect to either a conclusion of the structure of the source of a cut hyperedge, which is uniquely defined as the occurrence of formula $c$ such that $d$ is above $c$ and $c$ is not source of either a $\otimes$ or a $\parr$; the \emph{address of $d$} is then defined as the term $\loc_\mathcal{S}(d,x)=p_c(\ploc_\mathcal{S}(d,x))$.
\end{definition}

\begin{notation}
Let $\mathcal{S}$ be a proof-structure. We write $\mathrm{Ax}(\mathcal{S})$ (resp. $\mathrm{Cut}(\mathcal{S})$) the set of axioms (resp. cut) hyperedges in $\mathcal{S}$. Given $e\in\mathrm{Ax}(\mathcal{S})$ ($e\in\mathrm{Cut}(\mathcal{S})$), we write $C_e^l$ and $C_e^r$ the left and right conclusions (resp. sources) of $e$ respectively.
\end{notation}

\begin{definition}[Vehicle]
\label{def:vehicle}
The \emph{vehicle} of proof-structure $\mathcal{S}$ is the constellation $\mathcal{S}^{\mathrm{ax}}$ defined by: $|\mathcal{S}^{\mathrm{ax}}| = \mathrm{Ax}(\mathcal{S})$ and $\alpha_{\mathcal{S}^{\mathrm{ax}}}(e) = \gstar{\loc_{\mathcal{S}}(C_e^l), \loc_{\mathcal{S}}(C_e^r)}$.
\end{definition}

\begin{definition}[Colouration of constellation]
\label{def:colouration}
The colouration $\pm c.\Sigma$ of $\Sigma$ with a colour $c \in \mathcal{C}$ and a polarity $\pm \in \{+, -\}$ changes all the rays of its stars by the following functions $\phi_c^\pm$ defined by: $\phi_c^\pm(+d.t) = \pm c.t \quad \phi_c^\pm(-d.t) = \pm c.t \quad \phi_c^\pm(t) = \pm c.t$.
\end{definition}

\begin{notation}
Let $\Sigma_1,\Sigma_2$ be two constellations. We write $\Sigma_1 \cdot \Sigma_2$ their multiset union, i.e. the coproduct $\rho_{\Sigma_1}+\rho_{\Sigma_2}:\abs{\Sigma_1}+\abs{\Sigma_2}\rightarrow \setofstars$.
\end{notation}

\begin{definition}[Translation of proof-structures]
\label{def:transpn}
The translation of a proof-structure $\mathcal{S}$ into a constellation is defined as $\mathcal{S}^\bigstar = (\mathcal{S}^\mathrm{ax}, \mathcal{S}^\mathrm{cut})$ where $\mathcal{S}^\mathrm{cut}$ is the constellation defined by $|\mathcal{S}^\mathrm{cut}| = \mathrm{Cut}(\mathcal{S})$ and $\alpha_{\mathcal{S}^\mathrm{cut}}(e) = \gstar{p_{C_e^l}(x), p_{C_e^r}(x)}$.
Execution is extended to the translation of proof-structures: $\exec(\mathcal{S}^\star) = \exec (\posc{+c.\Sigma_{ax}} \cdot \negc{-c.\Sigma_{cut}})$.
\end{definition}

\begin{lemma}
\label{def:cutelimlemma}
Let $\mathcal{R}$ be a MLL proof-structure reducing in one step to another proof-structure $\mathcal{S}$. $\exec(\mathcal{R}^\bigstar)=\exec(\mathcal{S}^\bigstar)$.
\end{lemma}

\begin{proof}
Since each stars are unique by definition of $\loc$ the connexions are never ambiguous in the proof. We have two cases of reduction. If we have a cut on an axiom, the corresponding execution will be $\exec(\mathcal{R}^\bigstar) = \exec(\Sigma_\mathcal{R} + \vaxc{A_1}{t}{A_2^\bot}{x} + \vcut{A_2^\bot}{A_3})$. The cut $\vcut{A_2^\bot}{A_3}$ connects $\vaxc{A_1}{t}{A_2^\bot}{x}$ and another axiom $\vaxc{A_3}{u}{A_4^\bot}{v}$. We can connect them all to produce the star $\theta\vaxc{A_1}{t}{A_4^\bot}{v}$ where $\theta = \{x \mapsto u\}$. We identified $A_1$ and $A_3$ two occurrences of $A$. This is exactly how $\mathcal{S}$ is translated thus $\exec(\mathcal{R}^\bigstar) = \exec(\mathcal{S}^\bigstar)$.
If we have a $\parr/\otimes$ cut, the corresponding executions are $\exec(\mathcal{R}^\bigstar) = \exec(\Sigma_\mathcal{R} + \vcut{A \otimes B}{A^\bot \parr B^\bot})$ and
$\exec(\mathcal{S}^\bigstar) = \exec(\Sigma_\mathcal{S} + \vcut{A}{A^\bot} + \vcut{B}{B^\bot})$. The $\parr$ and $\otimes$ vertices have both two premises translated into rays $\posc{+c.p_{A^\bot \parr B^\bot}(\mathtt{l} \cdot u)}, \posc{+c.p_{A^\bot \parr B^\bot}(\mathtt{r} \cdot u')}$, $\posc{+c.p_{A \otimes B}(\mathtt{l} \cdot t)}$ and $\posc{+c.p_{A \otimes B}(\mathtt{r} \cdot t')}$ of axiom stars in $\Sigma_\mathcal{R}$. In order to construct a diagram, the cut has to be duplicated twice to saturate these rays. Since a $\parr$ and $\otimes$ vertex disappear after the reduction, the previous axioms are relocalised into $\posc{+c.p_{A^\bot}(u)}, \posc{+c.p_{B^\bot}(u')}, \posc{+c.p_A(t)}$ and $\posc{+c.p_B(t')}$ in $\Sigma_\mathcal{S}$. The cuts $\vcut{A}{A^\bot}$ and $\vcut{B}{B^\bot}$ connect these rays in the same way as in $\Sigma_\mathcal{R}$. Since $\Sigma_\mathcal{S}$ is $\Sigma_\mathcal{R}$ with a relocalisation of axioms they both have the same free rays. Moreover, all connexions with cuts does not involve free rays, hence $\exec(\mathcal{R}^\bigstar)=\exec(\mathcal{S}^\bigstar)$.
\end{proof}

\begin{theorem}[Dynamics]
\label{def:dynamics}
For a proof-net $\mathcal{R}$ of normal form $\mathcal{S}$, we have $\exec(\mathcal{R}^\bigstar)=\mathcal{S}^\bigstar$.
\end{theorem}
\begin{proof}
This result is a consequence of lemma \ref{def:cutelimlemma} by induction of the number of cut-elimination steps from $\mathcal{R}$ to $\mathcal{S}$, as well as the fact that $\exec(\mathcal{S}^\bigstar)=\mathcal{S}^\bigstar$ since $\mathcal{S}$ does not contain cuts.
\end{proof}

\begin{example}
\label{ex:mll}
Take the following reduction $\mathcal{S} \leadsto \mathcal{S}'$ of proof-structure:

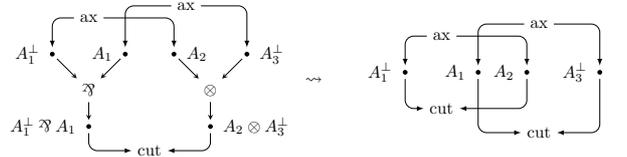
\begin{center}
\scalebox{0.64}{
\begin{minipage}{0.5\textwidth}
\begin{tikzpicture}
    \node[dot, label={180:$A_1^\bot$}] at (-0.75, 0.75) (a) {};
    \node[dot, label={180:$A_1$}] at (0.75, 0.75) (b) {};
    \node at (0, 0) (par) {$\parr$};
    \node[dot, label={180:$A_1^\bot \parr A_1$}] at (0, -0.75) (ab) {};
    \draw[-stealth] (a) -- (par);
    \draw[-stealth] (b) -- (par);
    \draw[-stealth] (par) -- (ab);

    \node[dot, label={0:$A_2$}] at (1.75, 0.75) (c) {};
    \node[dot, label={0:$A_3^\bot$}] at (3.25, 0.75) (d) {};
    \node at (2.5, 0) (tens) {$\otimes$};
    \node[dot, label={0:$A_2 \otimes A_3^\bot$}] at (2.5, -0.75) (cd) {};
    \draw[-stealth] (c) -- (tens);
    \draw[-stealth] (d) -- (tens);
    \draw[-stealth] (tens) -- (cd);

    \node at (1.25, -1.25) (cut) {cut};
    \draw[-latex, rounded corners=5pt] (ab) |- (cut);
    \draw[-latex, rounded corners=5pt] (cd) |- (cut);
    \node at (0, 1.5) (ax) {ax};
    \draw[-latex, rounded corners=5pt] (ax) -| (a);
    \draw[-latex, rounded corners=5pt] (ax) -| (c);
    \node at (2, 1.75) (ax) {ax};
    \draw[-latex, rounded corners=5pt] (ax) -| (b);
    \draw[-latex, rounded corners=5pt] (ax) -| (d);
\end{tikzpicture}
\end{minipage}
$\leadsto\qquad$
\begin{minipage}{0.45\textwidth}
\begin{tikzpicture}
    \node[dot, label={180:$A_1^\bot$}] at (-0.75, 0.75) (a) {};
    \node[dot, label={180:$A_1$}] at (0.75, 0.75) (b) {};

    \node[dot, label={180:$A_2$}] at (1.75, 0.75) (c) {};
    \node[dot, label={180:$A_3^\bot$}] at (3.25, 0.75) (d) {};

    \node at (0, 0) (cut) {cut};
    \draw[-latex, rounded corners=5pt] (a) |- (cut);
    \draw[-latex, rounded corners=5pt] (c) |- (cut);
    \node at (2, -0.5) (cut) {cut};
    \draw[-latex, rounded corners=5pt] (b) |- (cut);
    \draw[-latex, rounded corners=5pt] (d) |- (cut);
    \node at (0, 1.5) (ax) {ax};
    \draw[-latex, rounded corners=5pt] (ax) -| (a);
    \draw[-latex, rounded corners=5pt] (ax) -| (c);
    \node at (2, 1.75) (ax) {ax};
    \draw[-latex, rounded corners=5pt] (ax) -| (b);
    \draw[-latex, rounded corners=5pt] (ax) -| (d);
\end{tikzpicture}
\end{minipage},} \end{center}
then $\exec(\mathcal{S}^\bigstar) = \exec(\vaxc{A_1^\bot \parr A_1}{\mathtt{l} \cdot x}
     {A_2 \otimes A_3^\bot}{\mathtt{r} \cdot x}+$
     $\vcut{A_1^\bot \parr A_1}{A_2 \otimes A_3^\bot}+
\vaxc{A_2 \otimes A_2^\bot}{\mathtt{l} \cdot x}
{A_1^\bot \parr A_1}{\mathtt{r} \cdot x})$ and $\exec(\mathcal{S}^\bigstar) = \emptyset$. This is equal to $\exec(\mathcal{S}'^\bigstar)$ as the only correct saturated diagram for $\mathcal{S}^\bigstar$ (made by duplicating the star
representing the cut) has no free rays. If we look at the following reduction $\mathcal{S} \leadsto^* \mathcal{S}'$ instead:

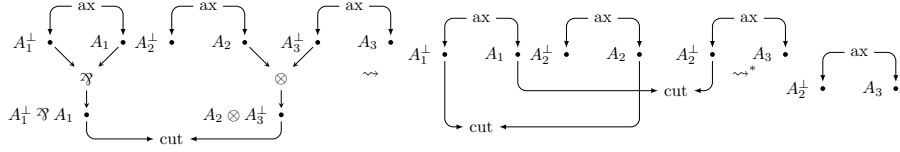
\begin{center}
\scalebox{0.64}{
\begin{minipage}{0.6\textwidth}
\begin{tikzpicture}
    \node[dot, label={180:$A_1^\bot$}] at (-0.75, 0.75) (a) {};
    \node[dot, label={180:$A_1$}] at (0.75, 0.75) (b) {};
    \node at (0, 0) (par) {$\parr$};
    \node[dot, label={180:$A_1^\bot \parr A_1$}] at (0, -0.75) (ab) {};
    \draw[-stealth] (a) -- (par);
    \draw[-stealth] (b) -- (par);
    \draw[-stealth] (par) -- (ab);

    \node[dot, label={180:$A_2^\bot$}] at (1.75, 0.75) (f1) {};
    \node[dot, label={180:$A_3$}] at (6.25, 0.75) (f2) {};

    \node[dot, label={180:$A_2$}] at (3.25, 0.75) (c) {};
    \node[dot, label={180:$A_3^\bot$}] at (4.75, 0.75) (d) {};
    \node at (4, 0) (tens) {$\otimes$};
    \node[dot, label={180:$A_2 \otimes A_3^\bot$}] at (4, -0.75) (cd) {};
    \draw[-stealth] (c) -- (tens);
    \draw[-stealth] (d) -- (tens);
    \draw[-stealth] (tens) -- (cd);

    \node at (1.75, -1.25) (cut) {cut};
    \draw[-latex, rounded corners=5pt] (ab) |- (cut);
    \draw[-latex, rounded corners=5pt] (cd) |- (cut);
    \node at (0, 1.5) (ax) {ax};
    \draw[-latex, rounded corners=5pt] (ax) -| (a);
    \draw[-latex, rounded corners=5pt] (ax) -| (b);
    \node at (2.5, 1.5) (ax) {ax};
    \draw[-latex, rounded corners=5pt] (ax) -| (f1);
    \draw[-latex, rounded corners=5pt] (ax) -| (c);
    \node at (5.5, 1.5) (ax) {ax};
    \draw[-latex, rounded corners=5pt] (ax) -| (f2);
    \draw[-latex, rounded corners=5pt] (ax) -| (d);
\end{tikzpicture}
\end{minipage}
$\leadsto\quad$
\begin{minipage}{0.55\textwidth}
\begin{tikzpicture}
    \node[dot, label={180:$A_1^\bot$}] at (-0.75, 0.75) (a) {};
    \node[dot, label={180:$A_1$}] at (0.75, 0.75) (b) {};

    \node[dot, label={180:$A_2^\bot$}] at (1.75, 0.75) (f1) {};
    \node[dot, label={180:$A_3$}] at (6.25, 0.75) (f2) {};

    \node[dot, label={180:$A_2$}] at (3.25, 0.75) (c) {};
    \node[dot, label={180:$A_2^\bot$}] at (4.75, 0.75) (d) {};

    \node at (0, -0.75) (cut) {cut};
    \draw[-latex, rounded corners=5pt] (a) |- (cut);
    \draw[-latex, rounded corners=5pt] (c) |- (cut);
    \node at (4, 0) (cut) {cut};
    \draw[-latex, rounded corners=5pt] (b) |- (cut);
    \draw[-latex, rounded corners=5pt] (d) |- (cut);
    \node at (0, 1.5) (ax) {ax};
    \draw[-latex, rounded corners=5pt] (ax) -| (a);
    \draw[-latex, rounded corners=5pt] (ax) -| (b);
    \node at (2.5, 1.5) (ax) {ax};
    \draw[-latex, rounded corners=5pt] (ax) -| (f1);
    \draw[-latex, rounded corners=5pt] (ax) -| (c);
    \node at (5.5, 1.5) (ax) {ax};
    \draw[-latex, rounded corners=5pt] (ax) -| (f2);
    \draw[-latex, rounded corners=5pt] (ax) -| (d);
\end{tikzpicture}
\end{minipage}
$\leadsto^*\quad$
\begin{minipage}{0.1\textwidth}
\begin{tikzpicture}
    \node[dot, label={180:$A_2^\bot$}] at (-0.75, 0.75) (a) {};
    \node[dot, label={180:$A_3$}] at (0.75, 0.75) (b) {};
    \node at (0, 1.5) (ax) {ax};
    \draw[-latex, rounded corners=5pt] (ax) -| (a);
    \draw[-latex, rounded corners=5pt] (ax) -| (b);
\end{tikzpicture}
\end{minipage}
} \end{center}
Then $\exec(\mathcal{S}^\bigstar) =
\vaxc{A_1^\bot \parr A_1}{\mathtt{l} \cdot x}
     {A_1^\bot \parr A_1}{\mathtt{r} \cdot x}
+
\vaxc{A_2^\bot}{x}{A_2 \otimes A_3^\bot}{\mathtt{l} \cdot x}
+$ \\
$\vaxc{A_2 \otimes A_3^\bot}{\mathtt{r} \cdot x}{A_3}{x}
+
\vcut{A_1^\bot \parr A_1}{A_2 \otimes A_3^\bot})$. Computing the executions, one obtains $\exec(\mathcal{S}^\bigstar) = \vaxc{A_2^\bot}{x}{A_3}{x} = \mathcal{S}'^\bigstar$.
\end{example}

\section{Interpreting the logical content of MLL}\label{sec:logical}

\subsection{Correctness of proof-structures}\label{subsec:correctionps}

As mentioned above, proof-structures are more permissive than sequent calculus proofs. In other words, some proof-structures do not represent proofs, and the syntax of MLL is therefore restricted to proof-nets, i.e. proof-structures that do represent sequent calculus proofs. A beautiful result of Girard, analysed by many subsequent works \cite{danos1989structure,danos1990logique,lafont1995proof,murawski2000dominator,de2011correctness}, is that those proof-structures that are proof-nets can be characterised by a topological property called a \emph{correctness criterion}. While Girard's original criterion, called the long-trip criterion \cite{girard1987linear}, is about the set of walks in a proof-structure, we will here work with Danos and Regnier's simplified criterion  \cite{danos1989structure}.

\begin{notation}
Given a proof-structure $\mathcal{S}=(V,E,s,t,\ell_V,\ell_E)$, we write $\parr(\mathcal{S})$ the subset $P\subseteq E$ of $\parr$-labelled edges, i.e. $\parr(\mathcal{S})=\{e\in E\mid \ell_E(e)=\parr \}$.
\end{notation}

We now define \emph{correction graphs}, which are the undirected hypergraphs obtained by removing one source of each $\parr$-labelled edge. The Danos-Regnier criterion then states that a proof-structure is a proof-net if and only if all correction graphs are trees.

\begin{definition}[Correction Graph]
\label{def:switching}\label{def:inducedgraph}
Let $\mathcal{S}=(V,E,s,t,\ell_V,\ell_E)$ be a proof-structure. A switching is a map $\sigma:\parr(\mathcal{S})\rightarrow\{l,r\}$. The correction hypergraph $\mathcal{S}_\sigma$ is the undirected hypergraph $(V,E,s')$ induced by the switching $\sigma$ is defined by letting $s'(e)=\{v\}\cup t(e)$ where $v$ is the left (resp. right) source of $v$ in $\mathcal{S}$ when $e\in\parr(\mathcal{S})$ and $\sigma(e)=l$ (resp. $\sigma(e)=r$), and $s'(e)=s(e)\cup t(e)$ for $e\not\in\parr(\mathcal{S})$.
\end{definition}

\begin{theorem}[Danos-Regnier correctness criterion \cite{danos1989structure}]
\label{thm:correctness}
A proof-structure $\mathcal{S}$ is a proof-net if and only if $\mathcal{S}_{\sigma}$ is a tree for all switching $\sigma$.
\end{theorem}

\begin{remark}
Each correction graphs can be defined as the union of two graphs: one which comes from the axioms and is uniquely defined by the proof-structure, and one which is obtained from edges that are not axioms and is dependent on the switching. This point of view allows for an \emph{interactive} formulation of the correctness criterion in which the set of axioms is \emph{tested} against graphs corresponding to switchings \cite{seiller2016computational}.
\end{remark}

\subsection{Reconstruction of correctness}\label{subsec:stellcorr}

We have already seen in the previous section how constellations can represent proofs. We now explain how to define tests to allow for an interactive, internal, representation of the correctness criterion. This is done by translating the Danos-Regnier criterion within the framework of stellar resolution.

We now use two colours $c$ (computation) and $t$ (testing). A vehicle will be coloured with the colour $c$ when we want its execution by connecting it with cuts and it will be coloured with the colour $t$ when being subject to tests against ordeals.

\begin{definition}[Ordeal]
\label{def:ordeals}
Let $\mathcal{S}$ be a MLL proof-structure and $\sigma$ one of its switchings. The ordeal $\mathcal{S}_\sigma^\bigstar$ associated to $\mathcal{S}_\sigma$ is the constellation obtained by translating all the vertices of $\mathcal{S}_\sigma^\bigstar$ in the following way:
\begin{itemize}
    \item $(C_e^d)^\bigstar = \gax{S}{C_e^d}$ for $e \in \mathrm{Ax}(\mathcal{S})$,
    \item $(A \parr_L B)^\bigstar = \gparl{A}{B}{A \parr B}$,
    \item $(A \parr_R B)^\bigstar = \gparr{A}{B}{A \parr B}$,
    \item $(A \otimes B)^\bigstar = \gtens{A}{B}{A \otimes B}$,
    \item We add $\gconc{A}$ for each conclusion $A$
\end{itemize}
We define $q_A(t)$ as a shortcut for $p_A(\mathtt{g} \cdot t)$ with $\mathtt{g}$ a constant only used for that definition.
\end{definition}

\begin{theorem}[Stellar correctness criterion]
\label{thm:correctness2}
A proof-structure $\mathcal{S}$ is a proof-net if and only if for all switching $\sigma$, we have $\exec(\post{+t.\mathcal{S}^\mathrm{ax}}\cdot \negc{-c.\mathcal{S}^\mathrm{cut}} \cdot \mathcal{S}_\sigma^\bigstar) = \gstar{p_{A_1}(x), ..., p_{A_n}(x)}$ where $A_1, ..., A_n$ are the conclusions of $\mathcal{S}$.
\end{theorem}

\begin{proof}[Proof (Sketch).]
$(\Rightarrow)$ One can observe that an ordeal $\mathcal{S}_\sigma^\bigstar$ perfectly reproduces the structure of $Sub = (V^{\mathcal{S}_\sigma}, E^{\mathcal{S}_\sigma}-\mathrm{Ax}(\mathcal{S}), s')$ which is always a forest (because it is made of the syntactic tree of $A_1, ..., A_n$) and so is $\ugraph{c,t}{(\negc{-c.\mathcal{S}^\mathrm{cut}} \cdot \mathcal{S}_\sigma^\bigstar)}$. We can construct a saturated diagram of $\negc{-c.\mathcal{S}^\mathrm{cut}} \cdot \mathcal{S}_\sigma^\bigstar$ by connecting all its stars and by cancelling the free rays not corresponding to conclusions thanks to the cuts in $\negc{-c.\mathcal{S}^\mathrm{cut}}$. This connexion corresponds to the contraction of a forest and we end up with the constellation $\Sigma_{ord} = \sigma_{A_1} + ... + \sigma_{A_n}$ where $p_{A_i}(x), \negt{-t.B_1(t_1)}, ..., \negt{-t.B_m(t_m)} \in \sigma_{A_i}$ for each $B_1, ..., B_m$ subformulas of $A_i$ which are conclusions of axioms. If $\mathcal{S}$ is a proof-net then $\mathcal{S}_\sigma$ must be tree. In this case, $(\ugraph{c,t}{\post{+t.\mathcal{S}^\mathrm{ax}}\cdot\Sigma_{ord})}$ must by acyclic (otherwise, we would have a cycle in $\mathcal{S}_\sigma$). Since all matchings are exact (i.e produce equations $t \doteq t$), we have a unique correct saturated diagram with free rays $p_{A_1}, ..., p_{A_n}$.
$(\Leftarrow)$ If $\exec(\post{+t.\mathcal{S}^\mathrm{ax}}\cdot \negc{-c.\mathcal{S}^\mathrm{cut}} \cdot \mathcal{S}_\sigma^\bigstar) = \gstar{p_{A_1}(x), ..., p_{A_n}(x)}$, it means that $\ugraph{c,t}{(\post{+t.\mathcal{S}^\mathrm{ax}}\cdot \negc{-c.\mathcal{S}^\mathrm{cut}} \cdot \mathcal{S}_\sigma^\bigstar)}$ is acyclic (otherwise we would have infinitely many correct saturated diagrams because only axioms cause cycles). Since the ordeal together with the translation of axioms are designed to reproduce the structure of $\mathcal{S}_\sigma$, then $\mathcal{S}_\sigma$ must by acyclic. Since the normalisation produces only one unique star, $\mathcal{S}_\sigma$ must be connected. Therefore, $\mathcal{S}_\sigma$ is a tree and $\mathcal{S}$ is a proof-net.
\end{proof}

\begin{corollary}[Corollary of the proof of Theorem \ref{thm:correctness2}]\label{cor:mixcorrection}
All correction graphs of a proof-structure $\mathcal{S}$ are acyclic if and only if $\post{+t.\mathcal{S}^\mathrm{ax}}\cdot \negc{-c.\mathcal{S}^\mathrm{cut}} \cdot \mathcal{S}_\sigma^\bigstar$ is strongly normalising.
\end{corollary}

\begin{example}
\label{ex:ordeals1}
We have two proof-structures for a chosen switching together with their corresponding ordeal\footnote{The purely esthetical use of the fractional notation is used to ease the reading of the stars.}:

\begin{minipage}{0.275\textwidth}
    \scalebox{0.75}{\begin{tikzpicture}
      \node[dot, label={180:$A$}] at (-0.75, 0.75) (a) {};
      \node[dot, label={180:$B$}] at (0.75, 0.75) (b) {};
      \node at (0, 0) (tens) {$\otimes$};
      \node[dot, label={180:$A \otimes B$}] at (0, -0.75) (ab) {};
      \draw[-] (a) -- (tens);
      \draw[-] (b) -- (tens);
      \draw[-] (tens) -- (ab);

      \node[dot, label={0:$A^\bot$}] at (1.25, 0.75) (ad) {};
      \node[dot, label={0:$B^\bot$}] at (2.75, 0.75) (bd) {};
      \node at (2, 0) (par) {$\parr_L$};
      \node[dot, label={180:$A^\bot \parr B^\bot$}] at (2, -0.75) (adbd) {};
      \draw[-] (bd) -- (par);
      \draw[-] (par) -- (adbd);

      \node at (0, 1.25) (ax1) {ax};
      \node at (2, 1.5) (ax2) {ax};
      \draw[-, rounded corners=5pt] (ax1) -| (a);
      \draw[-, rounded corners=5pt] (ax1) -| (ad);
      \draw[-, rounded corners=5pt] (ax2) -| (b);
      \draw[-, rounded corners=5pt] (ax2) -| (bd);
    \end{tikzpicture}}
\end{minipage}
\begin{minipage}{0.7\textwidth}
    $\dgloc{A \otimes B}{\mathtt{l} \cdot x}{A}+
     \dgloc{A \parr B}{\mathtt{l} \cdot x}{B}+
     \dgloc{A \otimes B}{\mathtt{r} \cdot x}{A^\bot}+
     \dgloc{A \parr B}{\mathtt{r} \cdot x}{B^\bot}+$ \\
    $\dgtens{A}{B}{A \otimes B}+\dgparl{A^\bot}{B^\bot}{A^\bot \parr B^\bot}+$ \\
    $\dgconc{A \otimes B}+\dgconc{A^\bot \parr B^\bot}$
\end{minipage}

\vspace{1em}
\begin{minipage}{0.2\textwidth}
    \scalebox{0.75}{\begin{tikzpicture}
      \node[dot, label={180:$A$}] at (-0.75, 0.75) (a) {};
      \node[dot, label={180:$A^\bot$}] at (0.75, 0.75) (b) {};
      \node at (0, 0) (tens) {$\otimes$};
      \node[dot, label={180:$A \otimes A^\bot$}] at (0, -0.75) (ab) {};
      \draw[-] (a) -- (tens);
      \draw[-] (b) -- (tens);
      \draw[-] (tens) -- (ab);

      \node at (0, 1.25) (ax) {ax};
      \draw[-, rounded corners=5pt] (ax) -| (a);
      \draw[-, rounded corners=5pt] (ax) -| (b);
  \end{tikzpicture}}
\end{minipage}
\begin{minipage}{0.7\textwidth}
    $\dgloc{A \otimes A^\bot}{\mathtt{l} \cdot x}{A}+
     \dgloc{A \otimes A^\bot}{\mathtt{r} \cdot x}{A^\bot}+
    \dgtens{A}{A^\bot}{A \otimes A^\bot}+\dgconc{A \otimes A^\bot}$
\end{minipage}

The first one will be connected with the axioms \[\vaxt{A \otimes B}{\mathtt{l} \cdot x}{A^\bot \parr B^\bot}{\mathtt{l} \cdot x}+\vaxt{A \otimes B}{\mathtt{r} \cdot x}{A^\bot \parr B^\bot}{\mathtt{r} \cdot x}.\] By connecting all stars of the ordeal, we will form a unique correct tree-like saturated diagram normalising into the conclusion stars $\gstar{p_{A \otimes B}(x)}+\gstar{p_{A^\bot \parr B^\bot}(x)}$ since they are the only free rays. From the second ordeal, when connected to $\vaxt{A \otimes A^\bot}{\mathtt{l} \cdot x}{A \otimes A^\bot}{\mathtt{r} \cdot x}$, we can construct infinitely many correct saturated diagrams because of the cycle. So the vehicle does not satisfy the stellar correctness criterion.
\end{example}

\subsection{Reconstruction of formulas}\label{subsec:formulas}

We here follow standard realisability constructions for linear logic \cite{girard2001locus,hyland2003glueing,seiller2012interaction}. Note that we explicit the \emph{trefoil property} \cite{seiller2016interaction} instead of the special case that is usually called \emph{adjunction}.

\begin{definition}[Orthogonality]
\label{def:ortho}
We say that two constellations $\Sigma_1, \Sigma_2$ are orthogonal w.r.t. a set of colours $\mathcal{A}\subseteq\mathcal{C}$, written $\Sigma_1 \perp_\mathcal{A} \Sigma_2$, when $\exec_{\mathcal{A}}(\Sigma_1 \cdot \Sigma_2)$ is strongly normalising. The orthogonal of a set of constellations is defined by $\mathbf{A}^{\bot_\mathcal{A}} = \{\Sigma \mid \forall \Sigma' \in \mathbf{A}, \Sigma \perp_\mathcal{A} \Sigma'\}$.
\end{definition}

\begin{notation}
Let $(P,\preceq)$ be a partially ordered set, and $A\subseteq P$. We write $\prefix{A}$ the set of \emph{prefixes in A} i.e $\prefix{A}=\{ a\in A \mid \forall b\in A, \lnot (b\preceq a) \}$.
\end{notation}

\begin{definition}[Order on rays]
We define the following partial order on the set of rays: given $r, r'$ two rays, $r\preceq r'$ if and only there exists a substitution $\theta$ such that $\theta r= r'$.
\end{definition}

We leave the verification that this defines a partial order to the reader. More intuitively, we have $r \preceq r'$ when $r$ is less specialized (thus more general) than $r'$. For instance, $\pm f(x) \preceq \pm f(g(y))$ because for $\theta = \{x \mapsto g(y)\}$, we have $\theta(\pm f(x)) = \pm f(g(y))$.

\begin{definition}[Location]
We define:
\begin{itemize}
\item the location $\location{\sigma}$ of a star $\rho_\sigma:\abs{\sigma}\rightarrow\rays{\mathbb{S}}$ as the set \[\prefix{\{\rho_\sigma(s) \mid s\in\abs{\sigma}\}};\]
\item the location $\location{\Sigma}$ of a constellation $\alpha_\Sigma: \abs{\Sigma}\rightarrow\setofstars$ as the set \[\prefix{\cup_{\sigma\in\abs{\Sigma}} \location{\alpha_\Sigma(\sigma)}};\]
\item the location $\location{\mathbf{A}}$ of a set $\mathbf{A}$ of constellations as the set $\prefix{\cup_{\Sigma\in \mathbf{A}} \location{\Sigma}}$.
\end{itemize}
\end{definition}

\begin{definition}[Conduct]
\label{def:conduct}
A set of constellation $\mathbf{A}$ is a conduct w.r.t. a set of colours $\mathcal{A}\subseteq\mathcal{C}$ if there exists a set of constellation $\mathbf{B}$ such that $\mathbf{A} = \mathbf{B}^{\bot_\mathcal{A}}$.
\end{definition}

\begin{proposition}[Biorthogonal closure]
A set of constellations $\mathbf{A}$ is a conduct w.r.t. a set of colours $\mathcal{A}\subseteq\mathcal{C}$ if and only if $\mathbf{A} = (\mathbf{A}^{\bot_\mathcal{A}})^{\bot_\mathcal{A}}$.
\end{proposition}

\begin{definition}[Intersection up to unification]
Let $R$ and $Q$ be sets of rays. We define their \emph{intersection up to unification} as the set:
\[ R\ucap Q=\prefix{\{ m\in\rays{\mathbb{S}} \mid \exists r\in R, q\in Q, r\preceq m\text{ and }q\preceq m \}}. \]
We say that $R$ and $Q$ are \emph{disjoint} when $R\ucap Q=\emptyset$; by extension, we say that two sets of constellations $\mathbf{A},\mathbf{B}$ are disjoint when $\location{\mathbf{A}}\ucap\location{\mathbf{B}}=\emptyset$.
\end{definition}

\begin{definition}[Tensor]
\label{def:tensor}
Let $\mathbf{A}, \mathbf{B}$ be disjoint conducts. We define their tensor by
\[\mathbf{A} \otimes_\mathcal{A} \mathbf{B} = (\{ \Sigma_1 \cdot \Sigma_2 \mid \Sigma_1 \in \mathbf{A}, \Sigma_2 \in \mathbf{B} \}^{\bot_\mathcal{A}})^{\bot_\mathcal{A}}.\]
\end{definition}

\begin{proposition}[Associativity/commutativity]
Given $\mathbf{A}, \mathbf{B}, \mathbf{C}$ pairwise disjoint conducts w.r.t. a set of colours $\mathcal{A}\subseteq\mathcal{C}$, we have $\mathbf{A} \otimes_\mathcal{A} \mathbf{B} = \mathbf{B} \otimes_\mathcal{A} \mathbf{A}$ and $\mathbf{A} \otimes_\mathcal{A} (\mathbf{B} \otimes_\mathcal{A} \mathbf{C}) = (\mathbf{A} \otimes_\mathcal{A} \mathbf{B}) \otimes_\mathcal{A} \mathbf{C}$.
\end{proposition}

\begin{definition}[Par and linear implication]
\label{def:par}
Let $\mathbf{A}, \mathbf{B}$ be conducts w.r.t. a set of colours $\mathcal{A}\subseteq\mathcal{C}$. We define:
$\mathbf{A} \parr_\mathcal{A} \mathbf{B} = (\mathbf{A}^{\bot_\mathcal{A}} \otimes_\mathcal{A} \mathbf{B}^{\bot_\mathcal{A}})^{\bot_\mathcal{A}}$ and $\mathbf{A} \multimap_\mathcal{A} \mathbf{B} = \mathbf{A}^{\bot_\mathcal{A}} \parr_\mathcal{A} \mathbf{B}$.
\end{definition}

\begin{theorem}[Associativity of execution]
\label{thm:assocexec}
Choose a set of colours $\mathcal{A}\subseteq\mathcal{C}$. For constellations $\Sigma_1, \Sigma_2, \Sigma_3$ such that $\location{\Sigma_1}\ucap\location{\Sigma_2}\ucap\location{\Sigma_3}=\emptyset$, we have
\[\exec_\mathcal{A}(\Sigma_1 \cdot \exec_\mathcal{A}(\Sigma_2 \cdot \Sigma_3)) = \exec_\mathcal{A}(\exec_\mathcal{A}(\Sigma_1 \cdot \Sigma_2) \cdot \Sigma_3)\]
\end{theorem}

\begin{proof}
Since all constellations have disjoint locations, their stars cannot be connected together and we have $\saturatedstardiag{}(\Sigma_2 \cdot \Sigma_3) = \saturatedstardiag{}(\Sigma_2) \uplus \saturatedstardiag{}(\Sigma_3)$. These diagrams can't be connected to the ones of $\Sigma_1$ which has its own saturated diagrams. Hence, $\saturatedstardiag{}(\Sigma_1 \cdot \exec_\mathcal{A}(\Sigma_2 \cdot \Sigma_3)) = \saturatedstardiag{}(\Sigma_1) \uplus \saturatedstardiag{}(\Sigma_2) \uplus \saturatedstardiag{}(\Sigma_3)$. With the same reasoning, we obtain $\saturatedstardiag{}(\Sigma_1 \cdot \exec_\mathcal{A}(\Sigma_2 \cdot \Sigma_3)) = \saturatedstardiag{}(\exec_\mathcal{A}(\Sigma_1 \cdot \Sigma_2) \cdot \Sigma_3)$. Therefore, $\exec_\mathcal{A}(\Sigma_1 \cdot \exec_\mathcal{A}(\Sigma_2 \cdot \Sigma_3)) = \exec_\mathcal{A}(\exec_\mathcal{A}(\Sigma_1 \cdot \Sigma_2) \cdot \Sigma_3)$.
\end{proof}

\begin{theorem}[Trefoil Property]
\label{thm:adjunction}
Choose a set of colours $\mathcal{A}\subseteq\mathcal{C}$. For all constellations $\Sigma_1, \Sigma_2, \Sigma_3$ s.t. $\location{\Sigma_1}\ucap\location{\Sigma_2}\ucap\location{\Sigma_3}=\emptyset$:
\[\begin{array}{p{11cm}}
$\big( \Sigma_2 \perp_\mathcal{A} \Sigma_3 \text{ and } \Sigma_1 \perp_\mathcal{A} \exec_\mathcal{A}(\Sigma_2 \cdot \Sigma_3)\big)$\hfill\\
\hfill  if and only if \hfill \\ \hfill $\big( \Sigma_1 \perp_\mathcal{A} \Sigma_2 \text{ and } \exec_\mathcal{A}(\Sigma_1 \cdot \Sigma_2) \perp_\mathcal{A} \Sigma_3 \big)$
\end{array}\]
\end{theorem}

\begin{theorem}[Alternative linear disjunction]
\label{thm:altlin}
Let $\mathbf{A}, \mathbf{B}$ be two conducts w.r.t. a set of colours $\mathcal{A}\subseteq\mathcal{C}$. Then: $\mathbf{A} \multimap_\mathcal{A} \mathbf{B} = \{\Sigma_f \mid \forall\  \Sigma_a \in \mathbf{A}, \exec_\mathcal{A}(\Sigma_f \cdot \Sigma_a) \in \mathbf{B} \text{ and } \Sigma_f \perp_\mathcal{A} \Sigma_a \}$.
\end{theorem}

As described in several work by the second author \cite{seiller2012interaction,seiller2012logique,seiller2016interaction}, the trefoil property and the associativity of execution ensure that one can define a $\ast$-autonomous category with conducts as objects and vehicles as morphisms. Due to the lack of space, we chose to omit this result which do not require new proof techniques and involves lots of bureaucratic definitions to deal with locations. We instead prove full soundness and completeness results.

\subsection{Truth and Soundness}\label{subsec:truth}

We will now define the interpretation of MLL formulas, which depends on a \emph{basis of interpretation}, and then prove full soundness and full completeness for MLL+MIX, that is the system obtained by adding to MLL sequent calculus the MIX rule, corresponding to the axiom scheme $A\parr B\multimap A\otimes B$. It is known that the correctness criterion for the MIX rule consists in taking the Danos-Regnier correction graphs and checking for acyclicity (but not connectedness). The proof of Theorem \ref{thm:correctness2} shows how the orthogonality of a vehicle w.r.t. the ordeal coincides with MLL+MIX correction, since the constellation $\post{+t.\mathcal{S}^\mathrm{ax}}\cdot \negc{-c.\mathcal{S}^\mathrm{cut}} \cdot \mathcal{S}_\sigma^\bigstar$ is strongly normalisable if and only if the structure $\mathcal{S}$ is acyclic. This is the key ingredient in the proof of full completeness. We are still working on obtaining similar results for MLL without MIX, using the results of Section \ref{subsec:stellcorr}.

We chose to fix the set of colours used, and omit the subscript that appeared in the constructions of the previous section. We also use a notion of \emph{localised} formulas, following previous work of the second author \cite{seiller2012interaction,seiller2012logique,seiller2016interaction}: it is defined using the same grammar as MLL formulas, except that variables are of the form $X_i(j)$, where $j$ is a term (here constructed from unary symbols $\mathtt{r},\mathtt{l}$ and $p_A$ for all occurrences of formulas $A$) used to distinguish occurrences, and one expect each occurrence to appear at most once in a formula.

A \emph{basis of interpretation} is a function $\Phi$ associating to each integer $i \in \nat$ a conduct $\Phi(i)$ in such a way that the conducts $(\Phi(i))_{i\in\nat}$ are pairwise disjoint. We moreover fix a bijection $\varphi:\setofaddresses{S}\times\nat\rightarrow\nat$, where $\setofaddresses{S}$ is the set of addresses $\loc_\mathcal{S}(\_,x)$ -- i.e. the countable set of all terms of the form $p_A(t(x))$ where $A$ ranges over conclusions of $\mathcal{S}$ and $t$ is a term constructed from the unary symbols $\mathtt{l}$ and $\mathtt{r}$. While the interpretation depends on the choice of $\varphi$, we will not indicate it in the notation for the sake of clarity.

In the next definition, we use the substitutions $\theta_r$ and $\theta_l$ which are defined as the identity for all variables except for $x$, and are defined respectively by $\theta_r(x)=\mathtt{r}(x)$ and $\theta_l(x)=\mathtt{l}(x)$.

\begin{definition}[Interpretation of formulas in proof-structures]
Given $\Phi$ a basis of interpretation, and $A$ a MLL formula occurrence identified by a unique unary function symbol $p_A$(cf. Definition \ref{def:loc}). We define the interpretation $I_\Phi(A,t)$ along $\Phi$ and a term $t$ inductively:
\begin{itemize}
\item $I_\Phi(A,t) = \Phi(\varphi(t,i))$ when $A=X_i$;
\item $I_\Phi(A,t) = \Phi(\varphi(t,i))^\bot$ when $A=X_i^\bot$;
\item $I_\Phi(A\otimes B,t) = I_\Phi(A,u) \otimes I_\Phi(B,v)$ where $u=\theta_l(t)$ and $v=\theta_r(t)$;
\item $I_\Phi(A\parr B,t) = I_\Phi(A,u) \parr I_\Phi(B,v)$ where $u=\theta_l(t)$ and $v=\theta_r(t)$.
\end{itemize}
We then define $I_\Phi(A)$ as $I_{\Phi}(A,p_A(x))$. We extend the interpretation to sequents by letting $I_\Phi(\vdash A_1, ..., A_n) = I_\Phi(A_1) \parr ... \parr _\Phi(A_n)$.
\end{definition}

\begin{definition}[Proof-like constellations]
A constellation $\Sigma$ is proof-like w.r.t. a set of addresses $\mathfrak{A}=\{\mathfrak{l}_1,\dots,\mathfrak{l}_k\}$ if $\sharp\Sigma=\mathfrak{A}$ and $\Sigma$ consists of binary stars only.
\end{definition}

The proof of the following result is a simple induction, combined with Theorem \ref{def:dynamics} to ensure that $\exec(\mathcal{S}^\bigstar)$ is proof-like w.r.t. the set $\sharp\mathcal{S}$ of addresses of conclusions of axions in $\mathcal{S}$.

\begin{theorem}[Full soundness]
\label{thm:soundness}
Let $\mathcal{S}$ be a MLL+MIX proof-net of conclusion $\vdash \Gamma$ and $\Phi$ a basis of interpretation. We have $\exec(\mathcal{S}^\bigstar) \in I_\Phi(\vdash \Gamma)$, and $\exec(\mathcal{S}^\bigstar)$ is proof-like w.r.t. $\sharp\mathcal{S}$.
\end{theorem}

We now consider syntax trees of formulas as incomplete proof-structures, where axioms are missing. We can extend the notion of switching to those \emph{pre-proof-structures}, and define their ordeal (as ordeals are defined without considering axioms hyperedges\footnote{We adapt the first case of Definition \ref{def:ordeals} and introduce the stars for atoms, i.e. for vertices that are not the target of an hyperedge.}). This is extended to sequents and used in the next lemma: given a sequent $\vdash \Gamma$ one can consider \emph{switchings} $\sigma$ of $\vdash \Gamma$ and ordeals $(\vdash \Gamma)_\sigma^\star$. We also define $\sharp\Gamma$ as the set of addresses of occurrences of atoms of $\Gamma$ seen as a pre-proof-structure.

\begin{lemma}\label{lem:tests}
Let $\Phi$ be a basis of interpretation, $\vdash \Gamma$ a sequent, and $\sigma$ a switching of $\vdash \Gamma$. Then $\exec_{\{c\}}((\vdash \Gamma)_\sigma^\star) \in (I_\Phi(\vdash \Gamma))^\bot$.
\end{lemma}

\begin{proof}[Proof sketch.]
The proof is done by induction.
\begin{itemize}
\item If $\vdash \Gamma$ has only formulas $X_i$ or $X_i^\bot$, then there is a single switching (there are no $\parr$), and $\exec_{\{c\}}((\vdash \Gamma)_\sigma^\star)=\sum \gstar{-t.p_{X_i},p_{X_i}}$. Since an element $\Sigma\in I_\Phi(\vdash \Gamma)$ is necessarily strongly normalisable and this implies that $\Sigma\cdot\exec_{\{c\}}((\vdash \Gamma)_\sigma^\star)$ is strongly normalisable, this shows the result.
\item If $\vdash \Gamma = \vdash \Delta, A\parr B$, then a switching $\sigma$ of $\vdash \Gamma$ is a switching $\bar{\sigma}$ of $\vdash \Delta,A,B$ extended to the additional $\parr$ connective linking $A$ and $B$. It should be clear that $\exec_{\{c\}}((\vdash \Delta,A\parr B)_\sigma^\star)=\exec_{\{c\}}((\vdash \Delta,A,B)_{\bar{\sigma}}^\star)$. This shows the result, since $I_\Phi(\vdash \Delta,A\parr B)=I_\Phi(\vdash \Delta,A,B)$.
\item If $\vdash \Gamma = \vdash \Delta, A\otimes B$, a switching of $\vdash \Gamma$ is a switching of $\vdash \Delta, A, B$, and $\exec_{\{c\}}((\vdash \Delta,A\otimes B)_\sigma^\star)$ can be defined from $\exec_{\{c\}}((\vdash \Delta,A, B)_\sigma^\star)$ by colouring the terms starting by $p_A$ and $p_B$ with a fresh colour $+u$ to obtain a constellation $\Theta$ and considering $\exec_{\{u\}}(\Theta\cdot \gstar{-u.p_A(x), -u.p_B(x), p_{A\otimes B}(x)})$. Moreover, one can show that $I_\Phi(\vdash \Delta,A\otimes B)$ is generated (in the sense of bi-orthogonal closure) by a set of constellations $E$ in which no star connects locations of $A$ with locations of $B$. This shows the result since this implies that $\exec_{\{c\}}((\vdash \Delta,A\otimes B)_\sigma^\star)\in E^\bot$ and it is known that $E^\bot=E^{\bot\bot}$ in general.
\qedhere\end{itemize}
\end{proof}

\begin{theorem}[Full completeness]
\label{thm:completeness}
If a constellation $\Sigma \in I_\Phi(\vdash \Gamma)$ is proof-like w.r.t. $\sharp\Gamma$, there exists a MLL+MIX proof-net $\mathcal{S}$ of conclusion $\vdash \Gamma$ such that $\Sigma = \mathcal{S}^{\mathrm{ax}}$.
\end{theorem}

\begin{proof}
A proof-like constellation $\Sigma \in I_\Phi(\vdash \Gamma)$ w.r.t to $\sharp\Gamma$ can always be considered as the interpretation of a proof-structure with only axioms; we can then construct a proof-structure $\mathcal{S}$ by considering the union of the latter with the syntax forest of $\vdash \Gamma$. Since $\Sigma$ belongs to $I_\Phi(\vdash \Gamma)$, and for all switching $\sigma$ of $\vdash \Gamma$ (equivalently, of $\mathcal{S}$) the ordeal $(\vdash \Gamma)_\sigma^\star=\mathcal{S}^\star_\sigma$ is orthogonal to $\Sigma$, Corollary \ref{cor:mixcorrection} shows that $\mathcal{S}$ is acyclic, i.e. satisfies the correctness criterion for MLL+MIX.
\end{proof}

\section{Perspectives and future works}
\label{sec:conclusion}

While we have shown here how to reconstruct the multiplicative fragment of linear logic, an interpretation of exponential connectives should also be possible using stellar resolution. Extension to the MELL fragment will be particularly interesting since it allows for the interpretation of System F \cite{girard1972interpretation} and pure $\lambda$-calculus \cite{danos1990logique,regnier1992lambda}. In fact, Girard's first article on Transcendental Syntax \cite{girard2017transcendental} sketches some reconstruction of the exponentials, but relies on constructions from later articles.

To interpret the additive connectives $\oplus, \&$, a way to exclude or force some choices in the construction of diagrams. For this purpose, Girard third article on Transcendental Syntax \cite{girard2016transcendental} mentions some involved coherence relations between stars that is not completely satisfying. This idea was already properly developed in Seiller's PhD thesis \cite{seiller2012logique} in the setting of interaction graphs; an improved and extended account can be found in a recent article by Nguyen and Seiller \cite{nguyen2019coherent}. We can expect to build on the latter to interpret additive connectives.

Last, but not the least, the third article on Transcendental Syntax \cite{girard2016transcendental3} suggests to interpret the terms of first-order logic as multiplicative propositions and the equality as the linear equivalence. This extension to first-order logic is the initial motivation behind the present work, and the authors expect to provide a formal account of these ideas. This would provide computational content for first-order logic in the sense of the Curry-Howard correspondence, something new and fascinating that would open numerous applications.



\bibliography{references}

\appendix

\end{document}